\documentclass[11pt, englis]{article}
\usepackage[margin=1in]{geometry}
\usepackage{subdepth}
\usepackage{blkarray, bigstrut} %
\usepackage{multicol}
\usepackage{enumitem}
\usepackage{amsmath}
\usepackage{amssymb,tabu,thm-restate}
\usepackage{mathtools}
\usepackage{color}
\usepackage{calc}
\usepackage{tikz}
\usetikzlibrary{positioning,arrows,decorations.pathreplacing,shapes}
\usetikzlibrary{calc}
\usepackage{multirow}
\usepackage{boxedminipage}
\usepackage{xifthen}
\usepackage{tabularx}
\usepackage{xcolor}
\usepackage{setspace}
\usepackage{url}
\usepackage{hyperref}
\newcommand{\RN}[1]{%
  \textup{\uppercase\expandafter{\romannumeral#1}}%
}
\newcommand{\arc}[1]{%
  \xrightarrow[]{\mbox{\tiny \RN{#1}}}%
}
\usepackage{titlesec}
\titlespacing{\paragraph}{0pt}{2ex}{0.5em}

\usepackage{mathpazo}

\usepackage{caption, subcaption}

\usepackage{algorithm}
\usepackage{algpseudocode}
\usepackage{pstricks}
\usepackage{pstricks-add}
\usepackage{pst-node,pst-tree}
\allowdisplaybreaks
\usepackage{cleveref}

\newcommand{\ignore}[1]{}
\newlength{\bibitemsep}
\newlength{\bibparskip}
\let\oldthebibliography\thebibliography
\renewcommand\thebibliography[1]{%
	\oldthebibliography{#1}%
	\setlength{\parskip}{\bibitemsep}%
	\setlength{\itemsep}{\bibparskip}%
}

\newcommand{\R}{\mathbb{R}}

\newcommand{\M}{\mathcal{M}}

\newcommand{\RR}{\mathbb{R}}

\DeclareMathOperator{\id}{id}

\DeclareMathOperator{\perm}{perm}
\DeclareMathOperator{\supp}{supp}

\newcommand{\OPT}{OPT}
\usepackage{algpseudocode}
\algrenewcommand\algorithmicrequire{\textbf{Input:}}
\algrenewcommand\algorithmicensure{\textbf{Output:}}

\DeclareMathOperator{\argmin}{argmin}

\newtheorem{claim}{Claim}[section]

\newtheorem{theorem}{Theorem}

\newtheorem{definition}{Definition}
\newtheorem{lemma}{Lemma}

\newtheorem{fact}{Fact}
\newenvironment{proof}[1][]{\par \noindent {\bf Proof #1}\ }{\hfill$\Box$\par \vspace{11pt}}

\newcommand{\norm}[1]{\Vert#1\Vert}
\newcommand{\dotp}[2]{\langle u_{#1}, v_{#2}\rangle_S}
\newcommand{\dotu}[2]{\langle u_{#1}, u_{#2}\rangle_S}
\newcommand{\dotv}[2]{\langle v_{#1}, v_{#2}\rangle_S}
\newcommand{\normu}[1]{\Vert u_{#1}\Vert_S}
\newcommand{\normv}[1]{\Vert v_{#1}\Vert_S}

\usepackage{rotating}
\usepackage{framed}

\title{Efficient Determinant Maximization for All Matroids}

\author{Adam Brown \thanks{Georgia Tech, \{ajmbrown, aladdha6, msingh94\}@gatech.edu}\and Aditi Laddha\footnotemark[1]\and Madhusudhan Pittu\thanks{Carnegie Mellon University, mpittu@andrew.cmu.edu}\and Mohit Singh\footnotemark[1]}
\date{}

\begin{document}

\maketitle

\begin{abstract}
Determinant maximization provides an elegant generalization of problems in many areas, including convex geometry, statistics, machine learning, fair allocation of goods, and network design.  In an instance of the determinant maximization problem, we are given a collection of vectors $v_1,\ldots, v_n \in \R^d$, and the goal is to pick a subset $S\subseteq [n]$ of given vectors to maximize the determinant of the matrix $\sum_{i \in S} v_iv_i^\top$, where the picked set of vectors $S$ must satisfy some combinatorial constraint such as cardinality constraint ($|S| \leq k$) or matroid constraint ($S$ is a basis of a matroid defined on $[n]$). 

In this work, we give a combinatorial algorithm for the determinant maximization problem under a matroid constraint that achieves $O(d^{O(d)})$-approximation for any matroid of rank $r\geq d$. This complements the recent result of~\cite{BrownLPST22} that achieves a similar bound for matroids of rank $r\leq d$, relying on a geometric interpretation of the determinant. Our result matches the best-known estimation algorithms~\cite{madan2020maximizing} for the problem, which could estimate the objective value but could not give an approximate solution with a similar guarantee. Our work follows the framework developed by~\cite{BrownLPST22} of using matroid intersection based algorithms for determinant maximization. To overcome the lack of a simple geometric interpretation of the objective when $r \geq d$, our approach combines ideas from combinatorial optimization with algebraic properties of the determinant. We also critically use the properties of a convex programming relaxation of the problem introduced by~\cite{madan2020maximizing}.
\end{abstract}
\thispagestyle{empty}
\newpage

\setcounter{page}{1}

\section{Introduction}

In an instance of the determinant maximization problem, we are given a collection of vectors $v_1,\ldots, v_n \in \R^d$, and the goal is to pick a subset $S\subseteq [n]$ of given vectors to maximize the determinant of the matrix $\sum_{i \in S} v_iv_i^\top$. The set $S$ must satisfy additional combinatorial constraints such as cardinality constraint ($|S| \leq k$) or matroid constraint ($S$ is a basis of a matroid defined on $[n]$). The determinant maximization problem under matroid constraint provides a general framework to model problems in various areas, including convex geometry~\cite{khachiyan1995complexity,summa2015largest,Nikolov2015}, experimental design in statistics~\cite{Pukelsheim2006,Allen-Zhu17nearoptimal,SinghX18,nikolov2018proportional}, reliable network design~\cite{baras2009efficient, LiPYZ19}, fair allocation of goods~\cite{anari2016nash,anari2018nash}, sensor placements~\cite{Joshi2009}, and determinantal point processes~\cite{Kulesza2012}. 

The case when the rank of the matroid $r$ is greater than $d$ is important as it generalizes crucial applications like Nash Social Welfare, experimental design, and network design. Algorithmically, this setting presents an interesting challenge as the difficulty of approximation for some applications differs when $r \leq d$ and when $r \gg d$. For example, the problem of maximizing Nash social welfare (NSW) can be modeled as an instance of determinant maximization under matroid constraints. While NSW is APX-hard~\cite{lee2017apx} in general, when the rank of the matroid constraint for NSW is at most dimension $d$, NSW is solvable in polynomial time. 

We describe a couple of applications that are modeled by determinant maximization problems. 

\paragraph{Nash Social Welfare problem.} The goal of Nash Social Welfare problem is to find an allocation of $m$ indivisible items to $d$ players to maximize welfare. Player $i$ has valuation $u_i(j)$ for item $j$, and these valuations are additive, i.e., for any $S \subseteq [m]$, $u_i(S) = \sum_{j\in S} u_i(j)$. The Nash Social Welfare objective of an allocation $\sigma$ is defined as the geometric mean of valuations of the agents, $NSW(\sigma) = (\prod_{i=1}^n \sum_{\sigma(j) = i} u_i(j))^{1/d}$. The geometric mean aims to capture both the fairness aspect, that each player individually obtains a bundle of large value, as well as total efficiency that the items go to players who value them more. \cite{anari2016nash} showed that this problem could be modeled as a particular case of determinant maximization where the constraint matroid is a partition matroid.  For player $i$ and item $j$, consider a vector $v_{i,j} = \sqrt{u_i(j)} \, e_i$ where $e_1, \ldots, e_d$ are the standard basis vectors in $\R^d$. Consider the partition matroid $\mathcal{M}$ on $[d] \times [m]$ with $[m]$ partitions, $P_j = \{(i, j): i \in [d]\}$ for $j \in [m]$. The independent sets of $\mathcal{M}$ are the sets which that at most one vector from each partition, i.e., $\mathcal{I} = \{ S \subset [d] \times [m]: |S \cap P_i | \leq 1\}$.
Then any basis, $S$, corresponds to a feasible allocation of items, and $\det(\sum_{S\cap P_j = \{i\}} v_{i,j} v_{i,j}^\top) $ is equal to the Nash Social Welfare objective raised to the power $d$. Our results give a polynomial in $d$ approximation algorithm for maximizing Nash Social Welfare.

\paragraph{Experimental design.} One of the classical problems in statistics is the experimental design problem~\cite{Pukelsheim2006}. In this problem, the goal is to estimate an unknown parameter vector $\theta^{\star}\in \RR^d$ by taking a set of linear observations of the form $y_i=v_i^{\top} \theta^{\star}+\eta_i$ where $v_i\in \RR^d$ and $\eta_i$ is Gaussian noise. We can select the test vector $v_i$ from a set of candidate vectors, $\{v_1,\ldots, v_n\}$, but each observation has some associated cost. The goal of experimental design is to select a subset of $r$ (often $r \ll n$) vectors, $S$, from the candidate set $\{v_1,\ldots, v_n\}$ such that observations on $S$ maximize the accuracy of estimating $\theta^{\star}$ for some criterion. Under the Gaussian assumption on the noise, for a set of observations $S$, the maximum likelihood estimator for $\theta^{\star}$ is the ordinary least squares solution on $S$, i.e., $\hat{\theta}=\argmin_{\theta}(\sum_{i\in S}|v_i^{\top} \theta -y_i|^2)$.  A simple computation shows that $\hat{\theta}-\theta^{\star}$ is distributed as a $d$-dimensional Gaussian with covariance matrix proportional to $\left(\sum_{i\in S} v_i v_i^{\top}\right)^{-1}$. One optimality criterion, called $D$-optimality, is to minimize the determinant of this matrix, which corresponds to minimizing the volume of the confidence ellipsoid. This can be formulated as an instance of the determinant maximization problem where the constraint matroid is the uniform matroid with rank $r$. A natural generalization of this problem is to restrict the feasible sets of observations to independent sets of a general matroid $\mathcal{M}$. This captures constraints that arise in practice. For example, some observations should be mutually exclusive, or observations should be spread apart in time or location~\cite{thiery2022two}. Our algorithms gives a $O(d^{O(d)})$-approximation for $D$-optimal design under matroid constraints when the rank of the matroid is greater than $d$.

\paragraph{Network design.} Given a graph $G=(V,E)$, a network design problem involves selecting a subgraph $H=(V,F)$ with $F\subseteq E$ such that $H$ is \emph{well-connected} and $F$ satisfies some combinatorial constraints. An example of a combinatorial constraint on $F$ could be $|F|\leq k$ for some integer $k$. Li et al.~\cite{LiPYZ19} consider the notion of connectivity defined by maximizing the number of spanning trees in $H$. This measure of connectivity has found applications in communication networks, network reliability, and as a predictor of the
spread of information in social networks. We refer the reader to~\cite{LiPYZ19} for further details. Here we note that the problem can be easily captured as a determinant maximization problem due to Kirchoff's formula that relates the determinant of the Laplacian of a graph to the number of spanning trees in it.  

\subsection{Our Results and Contributions}

Our main contribution is to give an approximation algorithm for the determinant maximization problem when the rank of the matroid $r$ more than the dimension $d$.
\begin{theorem}\label{thm:main}
    There is a polynomial time algorithm which, given a collection of vectors $v_1,\ldots, v_n \in \R^d$ and a matroid $\mathcal{M} = ([n], \mathcal{I})$ of rank $r \geq d$, returns a set $S \in \mathcal{I}$ such that
    \begin{equation*}
         \det\left( \sum_{i \in S} v_iv_i^\top \right) = \Omega\left(\frac{1}{d^{O(d)}} \right) \max_{S^* \in \mathcal{I}}\det\left( \sum_{i \in S^*} v_iv_i^\top \right).
    \end{equation*} 
\end{theorem}

This matches the $d^{O(d)}$-approximate estimation algorithm of~\cite{madan2020maximizing}, which only gives an estimate of the optimum value and results in $d^{O(d^2)}$-approximation algorithm.  This result also complements the previous work~\cite{BrownLPST22} which gives an $O(r^{O(r)})$-approximation algorithm for the determinant maximization problem under a matroid constraint when the rank $r$ is at most $d$. 

Throughout the paper, for a subset $S \subseteq [n]$, we also use $S$ to denote the corresponding set of vectors $\{ v_i : i \in S \}$ as well as the matrix whose columns are vectors $v_i$ in $S$. We use $V_S$ to denote the matrix formed by vectors $v_i$ for $i \in S$. For a subset of vectors $S$ and vectors $u, v$, we use $S-v+u$ to denote the set $(S\backslash \{v\}) \cup \{u\}$. 

\paragraph{Technical Overview.} Our starting point is the local search algorithm of~\cite{BrownLPST22}, which gives an $r^{O(r)}$-approximation when $r \leq d$. The algorithm tries to find a basis $S$ of $\mathcal{M}$ to optimize the objective $\det(V_S^\top V_S)$, or equivalently the product of the top $r$-eigenvalues of $V_SV_S^\top$. Let $S$ be any basis of $\mathcal{M}$ with a non-zero objective. The algorithm builds a directed bipartite graph, $G(S)$, called the \emph{exchange graph} of $S$ with bipartitions given by $[n]\setminus S$ and $S$. $G(S)$ contains a \emph{forward} arc $u \rightarrow v$ for $u\notin S$ and $v\in S$ if the vectors in $S-v+u$ are spanning in the linear matroid, and a \emph{backward} arc from $v$ to $u$ if $S-v+u$ is a basis of $\M$. Instead of the traditional vertex weights, the exchange graph $G(S)$ has edge weights which measure the change in objective after a single swap. The algorithm then iteratively finds a cycle $C$ in $G(S)$ with a \emph{really} negative weight and updates the solution by swapping all elements along the cycle, i.e., $S \leftarrow S\triangle C$. While we also follow a similar outline, the approach runs in to significant challenges. Firstly, when the size of the picked set is at most $d$, the determinant $\det(V_S^\top V_S)$ is equal to the square of the $r$-dimensional volume of the  parallelepiped spanned by the vectors in $S$. So, the weight of a forward arc $u\rightarrow v$ is simply the ratio of the volume of the parallelepipeds spanned by vectors in $S-v+u$ and $S$. This ratio can be reduced to a simple form using orthogonal projections. When $r > d$, this geometric relation, and the corresponding arc weights are no longer meaningful. In what follows, we show how to modify the algorithm to work in our case. In particular, we take a more algebraic rather than geometric approach. In our analysis, we crucially utilize the sparsity guarantees of a convex programming relaxation~\cite{madan2020maximizing} to avoid approximation factors that depend on the rank of the matroid $r$.   

We now outline the approach in detail, identifying crucial places where we need to take a different approach. The first critical difference is the weights on the edges of the exchange graph. Following ~\cite{BrownLPST22}, a natural approach would be to assign weight on forward arc $u \rightarrow v$ with $u\notin S$, $v\in S$ in the exchange graph as the logarithm of the ratio of objective when we replace $v$ with $u$ in the current basis $S$. So, for any $u \notin S$ and $v \in S$,
$w(u \rightarrow v) = -\frac12\log\frac{\det\left( V_S V_S^\top - vv^\top + uu^\top \right)}{ \det(V_S V_S^\top)} $. While in the case when $r=d$, the $w(u\rightarrow v)$ can be interpreted nicely using the geometric equivalence to volumes of parallelepipeds, and they also appear as entries of a natural linear map. But in our case, we take a more algebraic approach as such a simple geometric interpretation of the determinant does not exist when $r > d$. 
Using the matrix determinant lemma, we have
\begin{align}
 w(u \rightarrow v) &= -\frac12\log\det\left( V_S V_S^\top - vv^\top + uu^\top \right) + \frac12\log\det(V_S V_S^\top) \notag \\
 &= \frac12\log\left( \underbrace{\left(u^\top (V_S V_S^\top)^{-1}v\right)^2}_{\RN{1}} + \underbrace{\left(1+ u^\top (V_S V_S^\top)^{-1} u\right)\cdot \left(1-v^\top (V_S V_S^\top)^{-1}v\right)}_{\RN{2}}\right). \label{eq:1}
\end{align}
We consider the two terms in equation~\eqref{eq:1} separately and define \emph{two} forward arcs, denoted by $u \arc{1} v$ and $u \arc{2} v$ in $G(S)$ for every $u \in [n]\backslash S$ and $v\in S$. We introduce separate weights for the two new types of edges so that, with weight $-\log\left( \vert u^\top (V_S V_S^\top)^{-1}v \vert\right)$ for arcs of type $\RN{1}$ and a weight of $-\log\sqrt{ (1 + u^\top (V_S V_S^\top)^{-1} u)(1-v^\top (V_S V_S^\top)^{-1}v)}$ for arcs of type $\RN{2}$.

The next step is to show that if the current solution $S$ is significantly suboptimal, then there exists a cycle $C$ in the exchange graph of $S$ with the weights defined via equation~\eqref{eq:1} having a significantly negative total weight. We call such a cycle an $f$-violating cycle (formally defined in Definition~\ref{def:f-violating}). When $r \leq d$, the weights of the arcs in cycle $C$ correspond to entries of a linear map between $S$ and the optimal solution. This linear map is used crucially to  imply the existence of such an $f$-violating cycle. But the existence of such a linear map relies on the fact that the chosen vectors in $S$ must be linearly independent but clearly, that is not the case when $r>d$. Instead, we show that the arc weights defined in equation~\ref{eq:1} are closely related to entries of a \emph{bilinear} map given by the inner product space induced by the matrix $(V_S V_S^\top)^{-1}$ (see Definition~\ref{def:bilinear}). The matrix $(V_S V_S^\top)^{-1}$ completely describes the change in the determinant of $S$ when adding or removing vectors from $S$ (see Lemma~\ref{lem:det_update}). Moreover, the change in determinant is precisely equal to the determinant of a matrix whose entries are from this inner product space and, therefore, bilinear in the vectors being added and removed. The bilinearity is crucial to relate arc weights in the exchange graph with the change in determinant. It allows us to use the Cauchy-Schwarz inequality and similar tools (see in Lemma~\ref{lem:inv}).

Our algorithm uses $f$-violating cycles to both identify a suboptimal solution and modify it to a better valued solution. Therefore the algorithm's approximation factor depends on the optimality gap between the current and the optimal solution up to which we can find an $f$-violating cycle. 
The next challenge is that we can only guarantee the existence of an $f$-violating cycle when the current solution $S$ is suboptimal compared to the optimal solution $S^*$ by a factor that depends exponentially on $|S\Delta S^*|$, which can be as large as $2r$. To control the size of the symmetric difference between the current and the optimal solution, we use a convex program introduced in~\cite{madan2020maximizing} (discussed in Section~\ref{sec:convex_program}) to obtain sparse support for the matroid as a preprocessing step. This step reduces the problem to the case when the symmetric difference between $S$ and $S^*$ can be at most $3d^2$, allowing us to prove a guarantee that depends only on the dimension $d$ and not on the rank of the matroid $r$. Naively using this guarantee gives an approximation factor exponential in $O(d^2)$. To obtain the appropriate dependence in $d$, we use properties of matroids to carefully construct an $f$-violating cycle $C$ with at most $2d$ arcs as long as our current solution $S$ is sub-optimal compared to the optimal solution $S^*$ by a factor of at least $d^{cd}$ for some constant $c$.

We now state our main lemma about the existence of $f$-violating cycles.
\begin{restatable}{lemma}{existence} \label{lem:existence_cycle}
Let the number of elements in the ground set of $\mathcal{M}$ be $r+k$, and define $f:\mathcal{Z}_+ \rightarrow \mathcal{Z}_+ $ as $f(i) = \max\{2, (i!)^{11}\}$.
For any basis $S$ of $\mathcal{M}$ with $\det(V_S V_S^\top) > 0$, if there exists a basis $T$ such that $\det(V_TV_T^\top) > \det(V_S V_S^\top)  \cdot d^{4d}\cdot k^d \cdot f(2d)$, then there exists a cycle $C$ with at most $2d$ arcs in $G(S)$ such that
\begin{equation*}
    \sum_{e \in C} w(e) \leq -\log(f(|C|/2))\,. 
\end{equation*}
\end{restatable} 
As discussed above, we use a convex program introduced in~\cite{madan2020maximizing} to ensure $k = O(d^2)$ in the above lemma.

After finding a minimal $f$-violating cycle $C$, we update $S \leftarrow S\triangle C$. Our final step is to prove that updating any basis $S$ with a minimal $f$-violating cycle, $C$ in $G(S)$, strictly increases the determinant while keeping the new solution independent in the matroid $\mathcal{M}$. 
\begin{restatable}{theorem}{main} \label{thm:main2}
    Let $C$ be a minimal $f$-violating cycle in $G(S)$ and let $T = S \triangle C$. Then $T$ is independent in $\mathcal{M}$ and  $\det(V_TV_T^\top) > 2\cdot \det(V_S V_S^\top)$.
\end{restatable}

If there are no type $\RN{2}$ arcs in a minimal $f$-violating cycle $C$, then our proof follows a similar argument as in~\cite{BrownLPST22}. The existence of type $\RN{2}$ arcs creates significant challenges. First, we show that $C$ can contain only one type $\RN{2}$ arc (see Lemma~\ref{lem:typeII}). Second, exchanging on a type $\RN{2}$ arc changes the inner product space associated with the solution predictably by introducing additional bilinear error terms. We use the minimality of cycle $C$ to bound these error terms so that the problem reduces to the case when there is no type $\RN{2}$ arc in $C$. Lemma~\ref{lem:existence_cycle} and Theorem~\ref{thm:main2} together guarantee that as long as the current solution $S$ is sub-optimal to a factor of $O(d^{O(d)})$ compared to the optimal solution, the algorithm will be able to find a cycle $C$ in $G(S)$ such that the determinant of $S\Delta C$ is strictly greater than the determinant $S$ by a factor of two. By iteratively finding and exchanging on such cycles, we obtain Theorem~\ref{thm:main}.

\subsection{Related Work}

\paragraph{Determinant Maximization Under Cardinality Constraints.} Cardinality constraints can be represented as a uniform matroid constraint, and even in this special case, the determinant maximization problem is NP-hard ~\cite{Welch1982}. There has been substantial work on approximation algorithms for this problem ~\cite{Khachiyan1996,SummaEFM15,Nikolov2015,Allen-Zhu17nearoptimal,SinghX18,madan2020maximizing}, and the best know approximation algorithms provide an $e^r$-approximation for $r \leq d$ ~\cite{Nikolov2015}, and an $e^d$-approximation when $r \geq d$ ~\cite{SinghX18}. Improved results are known when $r$ is significantly larger than $d$. Allen-Zhu et al.~\cite{Allen-Zhu17nearoptimal} give a $(1+\epsilon)^d$-approximation when $r\geq \frac{d}{\epsilon}^2$ using spectral sparsification methods. The same guarantee can also be obtained when $r\geq \frac{d}{\epsilon}$~\cite{MadanSTU19} using the local search method. 

\paragraph{Determinant Maximization Under Matroid Constraints.} With general matroid constraints, there are $e^{O(r)}$-estimation algorithms when $r\leq d$~\cite{NikolovS16,anari2017generalization,AnariLGV19} and a $\min\{e^{O(r)}, O\bigl(d^{O(d)}\bigr)\}$-estimation algorithm when $r\geq d$~\cite{madan2020maximizing}. These algorithms use a convex relaxation of the problem and a connection to real stable polynomials and only estimate the optimum value. They can be derandomized into deterministic algorithms, which provide an $e^{O(r^2)}$-approximation when $r\leq d$ and an $O(d^{O(d^3)})$-approximation ~\cite{madan2020maximizing} when $r \geq d$. 

Another approach uses a non-convex program to efficiently find an approximate solution ~\cite{Ebrahimi17} when the constraint matroid is either a partition or a regular matroid. For a partition matroid with parts of constant size and rank $r\leq d$, they provide a $c^r$-approximation for some constant $c > 0$, and for regular matroids of rank $d$, the approximation factor is a function of the size of the ground set of the matroid.

\subsection{Organization}
In section~\ref{sec:prelim}, we discuss a few useful linear algebra primitives, then formally define the exchange graph and discuss the convex programming relaxation of the problem from~\cite{madan2020maximizing}. Finally, we state our algorithm for determinant maximization under matroid constraints. In section~\ref{sec:existence}, we show the existence of an $f$-violating cycle under appropriate suboptimality conditions and prove Lemma~\ref{sec:existence}. In Section~\ref{sec:update}, we prove Theorem~\ref{thm:main2}, which guarantees that our objective increases by a factor of $2$ after updating along an $f$-violating cycle. In Appendix~\ref{appendix:intro}, we provide some of the definitions and preliminaries related to matroids along with the proof of Theorem~\ref{thm:main}. We state the omitted proofs and lemmas from Section~\ref{sec:existence} and Section~\ref{sec:update} in Appendices~\ref{appendix:existence} and~\ref{appendix:update}, respectively. In Appendix~\ref{appendix:perm}, we provide a bound on the permanent of a matrix arising from a minimal $f$-violating cycle (first proved in~\cite{BrownLPST22}), which is crucial for our algorithm's success.

\section{Preliminaries} \label{sec:prelim}
As stated in the introduction, for a subset $S \subseteq [n]$, we also use $S$ to denote the corresponding set of vectors $\{ v_i : i \in S \}$. We use $V_S$ to denote the matrix formed by vectors $v_i$ for $i \in S$. For a subset of vectors $S$ and vectors $u, v$, we use $S-v+u$ to denote the set $(S\backslash \{v\}) \cup \{u\}$. 
\subsection{Linear Algebra}
\begin{lemma}[Matrix determinant lemma]
Suppose $A$ is an invertible $d\times d$ and $U, V$ are $d \times \ell$ matrices. Then $ \det(A + UV^\top) = \det(A) \cdot \det(I_\ell + V^\top A U)$.
\end{lemma}
Throughout the paper, we use the following lemma to measure the change in determinant when updating the current basis $S$.
\begin{lemma}[Determinant Update] \label{lem:det_update}
Suppose $S$ is a basis of $\mathcal{M}$ which is also linearly spanning and $X, Y \subset [n]$ with $|X| = |Y| = \ell$ and $Y \subset S$. Let $T = (S \cup X)\backslash Y$. Then
\begin{equation*}
    \det(V_TV_T^\top) = \det(V_SV_S^\top) \cdot \det\left(I_{2\ell} + \begin{bmatrix} V_X(V_SV_S^\top)^{-1} V_X & -V_X(V_SV_S^\top)^{-1} V_Y \\
    V_Y(V_SV_S^\top)^{-1} V_X & -V_Y(V_SV_S^\top)^{-1} V_Y \end{bmatrix}\right).
\end{equation*}
\end{lemma}
\begin{proof}
Expanding $V_TV_T^\top$ gives
\begin{equation*}
    \det(V_{T}V_{T}^\top) = \det(V_SV_S^\top + V_XV_X^\top - V_YV_Y^\top) = \det\left(V_SV_S^\top + \begin{bmatrix} V_X & -V_Y\end{bmatrix} \begin{bmatrix} V_X & V_Y\end{bmatrix}^\top\right).
\end{equation*}
Using the matrix determinant lemma gives
\begin{align*}
    \det(V_TV_T^\top) &= \det(V_SV_S^\top) \cdot \det\left(I_{2\ell} + \begin{bmatrix} V_X^\top \\ V_Y^\top \end{bmatrix} (V_S V_S^\top)^{-1} \begin{bmatrix} V_X & -V_Y\end{bmatrix} \right) \\
    &= \det(V_SV_S^\top) \cdot \det\left(I_{2\ell} + \begin{bmatrix} V_X(V_SV_S^\top)^{-1} V_X & -V_X(V_SV_S^\top)^{-1} V_Y \\
    V_Y(V_SV_S^\top)^{-1} V_X & -V_Y(V_SV_S^\top)^{-1} V_Y \end{bmatrix}\right).
\end{align*}
\end{proof}
For any basis $S$ which is linearly spanning, Lemma~\ref{lem:det_update} implies that the matrix $(V_SV_S^\top)^{-1}$ completely characterizes the change in the determinant of $S$ via the inner product space $\R^d \times \R^d \rightarrow \R : (u,v) \rightarrow u^\top (V_SV_S^\top)^{-1} v$. 
For ease of notation, we formally define the inner product space induced by $(V_S V_S^\top)^{-1}$.
\begin{definition}\label{def:bilinear}
For any basis $S$ which is linearly spanning, we use $\langle \cdot , \cdot \rangle_S$ to denote the inner product induced by $(V_S V_S^\top)^{-1}$ and $\norm{\cdot}_S$ to denote the corresponding norm. For any $u, v \in \R^d$, $\langle u, v \rangle_S = u^\top (V_S V_S^\top)^{-1} v$ and $\norm{u}_S = \sqrt{\langle u, u \rangle_S}$. 
\end{definition}
\subsection{Exchange Graph}
We now formally define the exchange graph $G(S)$ and the two types of arc weights in $G(S)$. We also define cycle weights and the notion of $f$-violating cycles, which was first introduced in ~\cite{BrownLPST22}. 
\begin{definition} [Exchange graph]
    For a subset of vectors $S = \{v_1,\ldots, v_r\}$ such that $S$ is independent in $\mathcal{M}$ and $S$ is a linearly spanning in $\R^d$, we define the exchange graph of $S$, denoted by $G(S)$, as a bipartite graph, where the right-hand bipartition consists of vectors in $S$, and the left-hand bipartition consists of all the vectors not in $S$. There is an arc from $v \in S$ to $u \notin S$ if $S -v + u$ is independent in $\mathcal{M}$. There are two arcs, labeled $\RN{1}$ and $\RN{2}$, from a vertex $u \notin S$ to a vertex $v$ in $S$ in $S - v + u$ is a linear spanning set.
\end{definition}

We now formally define the weight functions for arcs of type $\RN{1}$ and type $\RN{2}$. 
\begin{definition}[Weight functions] \label{def:wts}
In the exchange graph $G(S)$, all the backward arcs, from a vector $v \in S$ to a vector $u \in [n] \backslash S$, have weight $0$. The two forward arcs from $u \in [n] \backslash S$ to $v \in S $ have weights
\begin{align*}
    w(u \arc{1} v) &= -\log |u^\top (V_S V_S^\top)^{-1} v| = -\log |\langle u, v \rangle_S|\,, \\
    w(u \arc{2} v) &= -\log \sqrt{(1 +u^\top (V_S V_S^\top)^{-1} u) \cdot (1 +v^\top (V_S V_S^\top)^{-1} v)}\\ &= -\log\sqrt{(1+\norm{u}^2_S)\cdot(1-\norm{v}^2_S)}\,.
\end{align*}
\end{definition}

The following lemma, which is a direct corollary of Lemma~\ref{lem:det_update}, gives some intuition behind the chosen definitions.

\begin{lemma}\label{lem:weight_intuition}
    Let $S$ be a solution with $\det(V_S V_S^\top) > 0$ and $u \notin S$. Then for every $v \in S$
    \[ \frac{\det(V_S V_S^\top - vv^\top + uu^\top)}{\det(V_S V_S^\top)} = \exp\left( -2\cdot w(u \arc{1} v)\right) + \exp\left(-2\cdot w(u \arc{2} v) \right). \]
\end{lemma}

To characterize the set of really negative cycles, we use the following function $f$.
\begin{definition}[Function $f$] We define $f: \mathbb{Z}_+ \rightarrow \mathbb{Z}_+$ as $f(1) = 2$, and $f(i) = (i!)^{11}$ for all $i \geq 2$.
\end{definition}

Now we redefine cycle weights, $f$-violating cycles, and minimal $f$-violating cycles with respect to the weights introduced above.

\begin{definition}[Cycle Weight]
    The weight of a cycle $C$ in $G(S)$ is defined as $w(C) = \sum_{e \in C} w(e)$. 
\end{definition}
We use $|C|$ to denote the number of arcs in cycle $C$ and $\mathrm{ver}(C)$ to denote the set of vertices in cycle $C$. For a basis $S$ and cycle $C$ in $G(S)$, we use $S\Delta C$ to denote the symmetric difference of $S$ and $\mathrm{ver}(S)$.

\begin{definition}[$f$-Violating Cycle]\label{def:f-violating}
    A cycle $C$ in $G(S)$ is called an $f$-violating cycle if 
    \[ w(C) < -\log f(|C|/2)\,.\]
\end{definition}

\begin{definition}[Minimal $f$-Violating Cycle]
    A cycle $C$ in $G(S)$ is called a minimal $f$-violating cycle if 
    \begin{itemize}
        \item $C$ is an $f$-violating cycle, and
        \item for all cycles $C'$ such that $\mathrm{ver}(C') \subsetneq \mathrm{ver}(C)$, $C'$ is not an $f$-violating cycle.
    \end{itemize}
\end{definition}
A minimal $f$-violating cycle in $G(S)$ (if one exists) can be found in polynomial time following the same method as~\cite{BrownLPST22}.

\subsection{Convex Program}\label{sec:convex_program}
To guarantee the existence of an $f$-violating cycle, Lemma~\ref{lem:existence_cycle} requires that the ground set of $\mathcal{M}$ has cardinality at most $r + \text{poly}(d)$. A priori, this assumption might seem very limiting. However,~\cite{madan2020maximizing} gave a convex relaxation of the problem, which admits an $r + O(d^2)$ sparse solution. For the constraint matroid $\mathcal{M} = ([n], \mathcal{I})$, let $\mathcal{I}_s(\mathcal{M}) := \{S \in \mathcal{I} : |S| = s\}$ be the set of all independent sets of size $s$. We denote by $\mathcal{P}(\mathcal{M})$ the matroid base polytope of $\mathcal{M}$, which is the convex hull of the bases of $\mathcal{M}$. For any vector $z \in \R^n$ and a subset $S\subseteq [n]$, let $z(S) := \sum_{i\in S} z_i$. We let $\mathcal{Z} = \{z \in \R^n: \forall S \in \mathcal{I}_d(\mathcal{M}), z(S) \geq 0\}$. In~\cite{madan2020maximizing}, they introduce the optimization problem
\begin{equation}
     \sup_{x \in \mathcal{P}(\mathcal{M})} \inf_{z \in \mathcal{Z}} g(x,z) := \log\det \left( \sum_{i \in N} x_i e^{z_i} v_i v_i^\top \right).\tag{CP} \label{eq:cvx_program}
\end{equation}
The above program is a convex relaxation of the determinant maximization problem. Let $\OPT_{CP}$ denote the optimal value of the convex program and let $\OPT$ denote the optimal value of the determinant maximization problem on $\mathcal{M}$. We require the following fact about this convex program which can be found in ~\cite{madan2020maximizing}.

\begin{theorem}\label{thm:sparse-sol-gap}\cite{madan2020maximizing} Given vectors  $v_1, \ldots, v_n$ with constraint matroid $\mathcal{M} = ([n], \mathcal{I})$,
    there is a polynomial time algorithm that returns an optimal solution $\hat{x}$ to~\ref{eq:cvx_program} such that $|\supp(\hat{x})|\leq r + 2\left( \binom{d+1}{2} + d \right)$ where $\supp(\hat{x})$ is the set of variables with a non-zero value. Moreover,there exists a basis $T \subseteq \supp(\hat{x})$ such that
    \[ \det\left(\sum_{i\in T} v_iv_i^\top \right) \geq (2e^5 d)^{-d} \cdot \max_{S^* \in \mathcal{I}}\det\left( \sum_{i \in S^*} v_iv_i^\top \right). \]
\end{theorem}
Note that~\cite{madan2020maximizing}  does not explicitly address the existence of the set $T$. They instead provide a randomized rounding algorithm that achieves the determinant bound in expectation (See the proof of Theorem 2.3 in ~\cite{madan2020maximizing} for details), thus implying the existence of such a basis $T$. We only need the existential result for the success of our algorithm.

\subsection{Algorithm}
We begin by finding a sparse optimal fractional solution $\hat{x}$ to~\eqref{eq:cvx_program}. We then define a new matroid restricted to the support of $\hat{x}$. After this, our algorithm is structurally identical to the algorithm in~\cite{BrownLPST22}.

\begin{algorithm}[H]
\caption{Algorithm to find an approximation to $OPT$}
	\label{alg:exch}
\begin{algorithmic}
\Require $\mathcal{M} = ([n], \mathcal{I})$, $V = \{v_1, \ldots, v_n\}$, $\hat{x}$
\State $E \leftarrow \supp(\hat{x})$
\State $\mathcal{M}_E \leftarrow (E, \mathcal{I}_{|E|})$ \Comment{New matroid with sparse support and rank $r$}
\State $S \leftarrow$ basis of $\mathcal{M}_E$ with $\det(V_S V_S^\top) > 0$
\While{there exists an $f$-violating cycle in $G(S)$}
    \State $C =$ minimal $f$-violating cycle in $G(S)$
    \State $S = S\triangle C$
\EndWhile
\State Return $S$
\end{algorithmic}
\end{algorithm}
When we update $S \leftarrow S\triangle C$, Theorem \ref{thm:main2} guarantees that the new set is independent in the constraint matroid $\M$ and its determinant is strictly greater than that of $S$.

\section{Sparsity and Existence of a Short Cycle} \label{sec:existence}
In this section, we prove if the size of the ground set of $\mathcal{M}$ is $r+k$, and $S$ is a basis with suboptimality ratio $d^{4d}\cdot k^d\cdot f(2d)$, then $G(S)$ contains an $f$-violating cycle. Combining this with the sparsity guarantee of~\eqref{eq:cvx_program} ensures that an optimality gap of $d^{cd}$ (for some constant $c$) suffices for the existence of an $f$-violating cycle.

We start by proving there exists a basis $T$ such that the symmetric difference between $T$ and $S$ is $\ell$ and the ratio of the determinants of $T$ and $S$ is at least $(2\ell!)^{12}$, then $G(S)$ contains an $f$-violating cycle (see Lemma~\ref{lem:cyc_det}). A crucial ingredient of this proof is to relate the inner product space induced by $(V_SV_S^\top)^{-1}$ to arc weights in $G(S)$. With this fact in hand, we use the augmentation property of matroids to construct a basis that differs from $S$ in only $2d$ elements to complete the proof of Lemma~\ref{lem:existence_cycle}.

We restate Lemma~\ref{lem:existence_cycle} for the reader's convenience.
 
\existence*
\begin{proof}
Let $T = \{u_1, u_2, \ldots, u_r\}$ and $S = \{v_1, v_2, \ldots, v_r\}$ such that $S - v_i + u_i \in \mathcal{I}$ for all $i \in [r]$. By Fact~\ref{fact:basis_exch}, such an ordering always exists. Using the Cauchy-Binet formula,
\begin{equation}
    \frac{\det(V_TV_T^\top)}{\det(V_S V_S^\top)} = \frac{\sum_{Y \subset T, |Y| = d} \det(V_YV_Y^\top)}{\det(V_S V_S^\top)} = \sum_{Y \subset T, |Y| = d} \det(V_Y^\top (V_S V_S^\top)^{-1} V_Y)\,. \label{eq:2}
\end{equation}
Define $ U := T \backslash S$. Since the ground set contains $r+k$ elements, $|U| \leq k$. We partition the set of all $d$-subsets of $T$ by their intersection with $U$. For a set $W \subseteq U$, let $S_W = \{Y: Y \subset T, |Y| = d, Y \cap U = W\}$. 
Then
\begin{align*}
     \frac{\det(V_TV_T^\top)}{\det(V_S V_S^\top)} &= \!\!\!\sum_{Y \subset T, |Y| = d} \det(V_Y^\top (V_S V_S^\top)^{-1} V_Y) = \!\!\!\sum_{W \subseteq U, |W| \leq d} \sum_{Y \in S_W} \det(V_Y^\top (V_S V_S^\top)^{-1} V_Y)\,.
\end{align*}
The number of subsets of $U$ of size at most $d$ is $\sum_{i = 0}^d \binom{k}{d} \leq d\cdot (k/d)^d$. Therefore, there exists a $W \subseteq U$ with $|W| \leq d$ such that \begin{equation*}
    \sum_{Y \in S_W} \det(V_Y^\top (V_S V_S^\top)^{-1} V_Y) \geq  d^{d-1} \cdot k^{-d}\cdot \frac{\det(V_TV_T^\top)}{\det(V_S V_S^\top)} \geq (2d)! \cdot f(2d)\,,
\end{equation*}
where the last inequality follows from the hypothesis of the lemma.

Since $\{S \cap T\} \cup W \subset T$, by the downward closure property of matroids, $\{S \cap T\} \cup W$ is independent in $\mathcal{M}$. So we can extend $\{S \cap T\} \cup W $ to a basis, $T_1$, of $\mathcal{M}$ in $S \cup W$ such that $\{S \cap T\} \cup W \subseteq T_1$. Again, using the Cauchy-Binet formula on $T_1T_1^\top$ gives
\begin{align*}
    \frac{\det(V_{T_1}V_{T_1}^\top)}{\det(V_S V_S^\top)} &=  \sum_{Y \subset T_1, |Y| = d} \det(V_Y^\top (V_S V_S^\top)^{-1} V_Y) \geq \sum_{Y \subset\{S \cap T\} \cup W, |Y| = d} \det(V_Y^\top (V_S V_S^\top)^{-1} V_Y)\\
    &\geq \sum_{Y \in S_W, |Y| = d} \det(V_Y^\top (V_S V_S^\top)^{-1} V_Y) \geq (2d)! \cdot f(2d).
\end{align*}
Since $|T_1\backslash S| \leq d$, using Lemma~\ref{lem:cyc_det}, there exists an $f$-violating cycle in $G(S)$.
\end{proof}

\begin{lemma} \label{lem:cyc_det}
Let $T, S$ be two bases of $\mathcal{M}$ such that $|T\backslash S| = \ell$ and $\det(V_TV_T^\top) \geq \det(V_S V_S^\top) \cdot (2\ell)! \cdot f(2\ell)$. Then there exists an $f$-violating cycle in $G(S)$.
\end{lemma}
\begin{proof}
Let $X := T \backslash S$ and $Y := S\backslash T$. Using Lemma~\ref{lem:det_update}, we have
\begin{align*}
    \frac{\det(V_TV_T^\top)}{\det(V_S V_S^\top)} 
    &= \det\left( \begin{bmatrix}
   I_\ell + V_X^\top (V_S V_S^\top)^{-1}  V_X & -V_X^\top (V_S V_S^\top)^{-1}  V_Y \\
    -V_Y^\top (V_S V_S^\top)^{-1}  V_X & I_\ell - V_Y^\top (V_S V_S^\top)^{-1}  V_Y
    \end{bmatrix} \right) .
\end{align*}
Let $A := \begin{bmatrix}
   I_\ell + V_X^\top (V_S V_S^\top)^{-1}  V_X & -V_X^\top (V_S V_S^\top)^{-1} V_Y \\
    -V_Y^\top (V_S V_S^\top)^{-1}  V_X & I_\ell - V_Y^\top (V_S V_S^\top)^{-1}  V_Y
    \end{bmatrix} $ and $a_{i,j} := [A]_{i,j}$.
Expanding the determinant of $A$ gives
\begin{align*}
    |\det(A)| = |\sum_{\sigma \in \mathcal{S}_{2\ell}} \prod_{i=1}^{2\ell} a_{i,\sigma(i)} | \leq \sum_{\sigma \in \mathcal{S}_{2\ell}} \prod_{i=1}^{2\ell} |a_{i,\sigma(i)}|.
\end{align*}
Since $|\mathcal{S}_{2\ell}| = (2\ell)!$, there exists a permutation $\tau \in \mathcal{S}_{2\ell}$ such that 
\begin{equation}
    \prod_{i=1}^{2\ell} |a_{i,\tau(i)}| \geq \frac{1}{(2\ell)!} \cdot \frac{\det(V_TV_T^\top)}{\det(V_S V_S^\top)} \geq f(2\ell)\,. \label{eq:3}
\end{equation}
We now relate $|a_{i, \tau(i)}|$ to the weight of an arc in $G(S)$ for every $i$ and then use those arcs to construct a set of cycles, one of which will give an $f$-violating cycle.

To bound $|a_{i, \tau(i)}|$, consider the partition of the set of indices, $[2\ell]$, into four sets according to $\tau$ as follows. Let  $I_1 = \{i \in [\ell]: \tau(i) \leq \ell\}$, $I_2 = \{i \in [\ell]: \tau(i) > \ell\}$, $I_3 = \{i \in [\ell+1, 2\ell]: \tau(i) \leq \ell\}$, and $I_4 = \{i \in [\ell+1, 2\ell]: \tau(i) > \ell\}$ (see Figure \ref{fig:index_partition}). Note that $|I_1| = |I_4|$; the number of permutation entries in the top left block is equal to the number of permutation entries in the bottom right block. Let $X = \{u_1, u_2, \ldots, u_\ell\}$ and $Y =  \{v_{\ell+1}, v_{\ell+2}, \ldots, v_{2\ell}\}$ so that $S - v_{\ell+i} + u_i \in \mathcal{I}$ for all $i\in [\ell]$.

\begin{figure}
    \centering
    \includegraphics[width=0.6\textwidth]{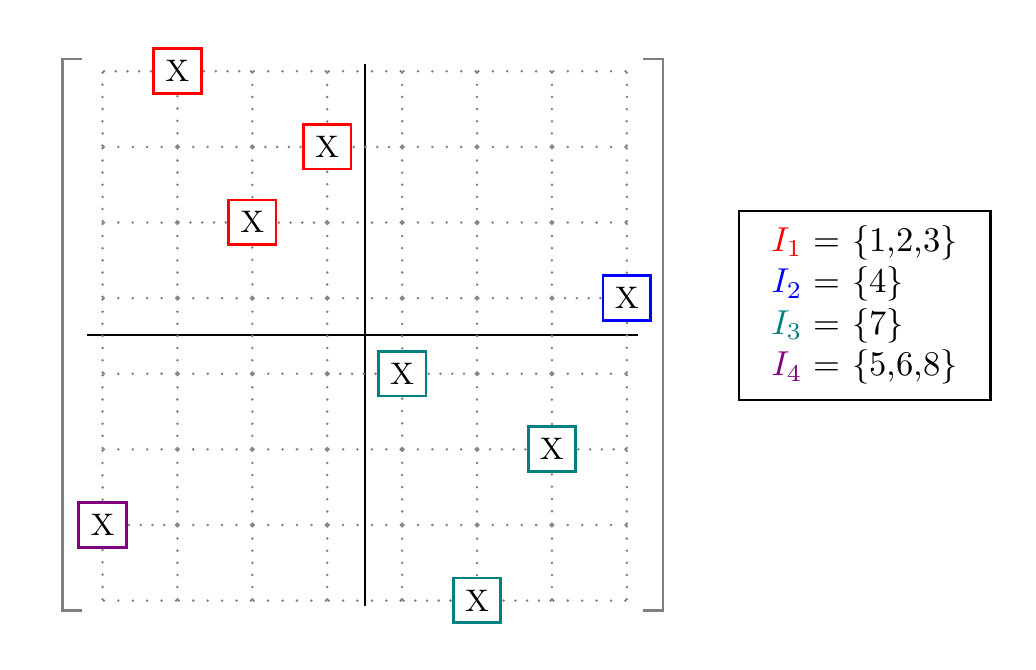}
    \caption{A permutation and corresponding index partition when $\ell = 4$.}
    \label{fig:index_partition}
\end{figure}

For any $i \in I_1$, 
\begin{align*}
    |a_{i, \tau(i)}| &\leq 1+|u_i^\top (V_S V_S^\top)^{-1} u_{\tau(i)}| = 1 + |\langle u_i, u_{\tau(i)} \rangle_S| \leq 1 + \norm{u_i}_S \cdot \norm{u_{\tau(i)}}_S \tag{by Cauchy Schwarz}\\
    &\leq \sqrt{(1+\norm{u_i}_S^2) \cdot (1+\norm{u_{\tau(i)}}_S^2)}\,.
\end{align*}
For any $i \in I_4$ with $i \neq \tau(i)$, by Lemma~\ref{lem:inv}, $|a_{i, \tau(i)}| = |\dotv{i}{\tau(i)}|\leq \sqrt{(1-\norm{v_i}_S^2)\cdot (1-\norm{v_{\tau(i)}}_S^2)}$.
This inequality is trivially true when $\tau(i) = i$ as $a_{i,i} = (1-\normv{i}^2)$ for all $i \in I_4$.

For  $i \in I_2$, $|a_{i,\tau(i)}| = |u_i^\top (V_S V_S^\top)^{-1} v_{\tau(i)}| = |\langle u_i, v_{\tau(i)}\rangle_S|$. Similarly, for $i \in I_3$, $|a_{i, \tau(i)}| = |\langle u_{\tau(i)}, v_{i}\rangle_S|$.

Putting it all together,
\begin{align*}
    \prod_{i=1}^{2\ell} |a_{i,\tau(i)} | 
    {}={}&  \left(\prod_{i\in I_1} |a_{i,\tau(i)} |\right)\cdot  \left(\prod_{i\in I_2} |a_{i,\tau(i)} |\right)\cdot  \left(\prod_{i\in I_3} |a_{i,\tau(i)} |\right)\cdot  \left(\prod_{i\in I_4} |a_{i,\tau(i)} |\right)\\
    {}\leq{}&  \left(\prod_{i\in I_1} \sqrt{(1+\norm{u_i}_S^2)\cdot (1+\norm{u_{\tau(i)}}_S^2)}\right)\cdot  \left(\prod_{i\in I_4} \sqrt{(1-\norm{v_i}_S^2)\cdot (1-\norm{v_{\tau(i)}}_S^2)}\right) \\
    & {}\cdot{}  \left(\prod_{i\in I_2} |\langle u_{i},v_{\tau(i)} \rangle_S |\right)\cdot  \left(\prod_{i\in I_3} |\langle u_{\tau(i)}, v_i \rangle_S |\right).
\end{align*}  
We will now relate the R.H.S. of the inequality to the weights of cycles in $G(S)$. For this purpose, we define a new permutation $\delta$ as follows.
Since $|I_1| = |I_4|$, first define a bijection $h: I_1 \rightarrow I_4$.
Then the permutation $\delta$ is given by $\delta(i)= \tau(h(i)) $ for every $i \in I_1$, $\delta(i)= \tau(h^{-1}(i)) $ for every $i \in I_4$, and $\delta(i)= \tau(i) $ for any $i \in I_2 \cup I_3$. Then the inequality becomes
\begin{align}
    \prod_{i=1}^{2\ell} |a_{i,\tau(i)} |
    {}\leq{}& \prod_{i\in I_1} \sqrt{(1+\norm{u_i}_S^2)\cdot (1-\norm{v_{\delta(i)}}_S^2)}\cdot  \prod_{i\in I_4} \sqrt{(1-\norm{v_{i}}_S^2)(1+\norm{u_{\delta(i)}}_S^2)}\cdot \notag \\
    & {}\cdot{} \prod_{i\in I_2} |\langle u_{i},v_{\delta(i)} \rangle_S|\cdot  \prod_{i\in I_3} |\langle u_{\delta(i)}, v_i \rangle_S|\,, \label{eq:4}
\end{align}
where we have swapped some terms using the new permutation to ensure that each square root has one addition and one subtraction term.
We claim that  $S + u_i - v_{\delta(i)}$  is a linear spanning subset of $\mathbb{R}^d$ for any $i \in I_1 \cup I_2$. To observe this, note that  $\det(V_S V_S^\top + u_iu_i^\top -v_{\delta(i)}v_{\delta(i)}^\top) = \det(V_S V_S^\top)\cdot (\dotp{i}{\delta(i)}^2 + (1+\normu{i}^2) \cdot (1-\normv{\delta(i)}^2)) $ from Lemma~\ref{lem:weight_intuition}. Since $ \prod_{i=1}^{2\ell} |a_{i,\tau(i)} | > 0$,  inequality~\eqref{eq:4} implies, either $|\dotp{i}{\delta(i)}| > 0$ or $(1+\norm{u_i}_S^2)\cdot (1-\norm{v_{\delta(i))}}_S^2) > 0$ for every $i \in I_1 \cup I_2$. Therefore $\det(V_S V_S^\top + u_iu_i^\top -v_{\delta(i)}v_{\delta(i)}^\top) > 0$ and $S + u_i - v_{\delta(i)}$ is a linearly spanning set of cardinality $r$.  So, $G(S)$ contains forward arcs $u_i \arc{1} v_{\delta(i)}$ and $u_i \arc{2} v_{\delta(i)}$ for all $i \in I_1 \cup I_2$. Similarly, $G(S)$ contains arcs $u_{\delta(i)} \arc{1} v_{i}$ and $u_{\delta(i)} \arc{2} v_{i}$ for every $i \in I_3 \cup I_4$.
So, we can rewrite equation~\eqref{eq:4} in terms of arc weights from $G(S)$.
\begin{align*}
     \prod_{i=1}^{2\ell} |a_{i,\tau(i)} |&\leq \exp(-\sum_{i\in I_1} w(u_i \arc{2} v_{\delta(i)}) - \sum_{i\in I_4} w(u_{\delta(i)} \arc{2} v_{i}) -  \sum_{i\in I_2 } w(u_i \arc{1} v_{\delta(i)}) - \prod_{i\in I_3 }w (u_{\delta(i)}\arc{1} v_i))\,.
\end{align*}
From~\eqref{eq:3}, we have $ \prod_{i=1}^{2\ell} |a_{i,\tau(i)} |\geq f(2\ell)$ and therefore
\begin{equation}
   \exp(-\sum_{i\in I_1} w(u_i \arc{2} v_{\delta(i)}) - \sum_{i\in I_4} w(u_{\delta(i)} \arc{2} v_{i}) -  \sum_{i\in I_2 } w(u_i \arc{1} v_{\delta(i)}) - \prod_{i\in I_3 }w (u_{\delta(i)}\arc{1} v_i)) \geq f(2\ell).  \label{eq:5}
\end{equation}
Consider a weighted bipartite graph $H(S)$ with bipartitions $X$ and $Y$, and forward arcs $u_i \arc{2} v_{\delta(i)}$ for all $i \in I_1$,
$u_i \arc{1} v_{\delta(i)}$  for all $i \in I_2$,
 $u_{\delta(i)} \arc{1} v_i$ for all $i \in I_3 $,
$u_{\delta(i)} \arc{2} v_i$ for all $i \in I_4 $. Additionally, for each $i \in [\ell]$, we add two backward arcs from $v_{\ell + i} \rightarrow u_i$ with weight $0$ to $H(S)$. Since $S - v_{\ell+i} + u_i \in \mathcal{I}$, $v_{\ell + i} \rightarrow u_i$ is also an arc in $G(S)$. So the arcs in $H(S)$ are a multiset of arcs in $G(S)$.

$H(S)$ contains $4\ell$ arcs with every vertex incident to exactly two incoming arcs and two outgoing arcs. Therefore $H(S)$ is an Eulerian graph, and we can decompose it into arc disjoint cycles $C_1, \ldots, C_k$. So, summing over the weights of arcs in $H(S)$ gives
\begin{align*}
    \sum_{i\in I_1} &w(u_i \arc{2} v_{\delta(i)}) + \sum_{i\in I_4} w(u_{\delta(i)} \arc{2} v_{i})+ \sum_{i\in I_2 }w(u_i \arc{1} v_{\delta(i)})+\sum_{i\in I_3 } w(u_{\delta(i)} \arc{1} v_i) = \sum_{i=1}^k w(C_i).
\end{align*}
Taking exponential and using~\eqref{eq:5}, we get $  \prod_{i=1}^k \exp(-w(C_i)) \geq f(2\ell)$.
Since $C_1, \ldots, C_k$ are arc disjoint and $H(S)$ contains $4\ell$ arcs, $\sum_{i=1}^k|C_i| = 4\ell$. Now using the fact that $f(a)\cdot f(b) \leq f(a+b)$, we get 
\begin{align*}
     \prod_{i=1}^k \exp(-w(C_i)) \geq f(2\ell) \geq \sum_{i=1}^k f(|C_i|/2). 
\end{align*}
So there exists a cycle $C$ in $C_1,\ldots, C_k$ such that $\exp(-w(C)) > f(|C|/2)$. Since every arc in $H(S)$ is also an arc in $G(S)$, the cycle $C$ is present in $G(S)$ too. Therefore $C$ is an $f$-violating cycle in $G(S)$.
\end{proof}

\section{Short Cycles give Large Improvement} \label{sec:update}
In this section, we prove that exchanging on a minimal $f$-violating cycle in $G(S)$ gives a basis whose determinant is strictly larger than $S$. We first show that any minimal $f$-violating cycle contains at most one edge of type $\RN{2}$ (see Lemma~\ref{lem:typeII}). Using this observation, we bound the change in determinant in two steps.

First, if the minimal $f$-violating cycle, $C$, only contains type $\RN{1}$ edges, then the proof follows the outline of Theorem 2 in~\cite{BrownLPST22}. Let $T = S\Delta C$, then
\[\det(V_T V_T^\top) \geq \det(V_SV_S^\top)\cdot \det(V_{T\backslash S}^\top (V_SV_S^\top)^{-1} V_{S\backslash T})^2,\]
(see Lemma~\ref{lem:local}).  Every non-zero entry of $V_{T\backslash S}^\top (V_SV_S^\top)^{-1} V_{S\backslash T}$ corresponds to the weight of a type $\RN{1}$ arc whose endpoints lie in the vertices of $C$. Using the minimality of $C$, the problem of bounding $|V_{T\backslash S}^\top (V_SV_S^\top)^{-1} V_{S\backslash T}|$ can be reduced to bounding the determinant of a numerical matrix whose entries depend only on the function $f$ and the number of edges in $C$, $|C|$. Using properties of the function $f$, we prove that when $C$ is a minimal $f$-violating cycle with only type $\RN{1}$ edges,
\begin{equation}
    |\det(V_{T\backslash S}^\top (V_SV_S^\top)^{-1} V_{S\backslash T})| \gg 1\,. \label{eq:6}
\end{equation}

Otherwise, if $C$ contains exactly one type $\RN{2}$ edge $e = (u\arc{2} v)$, we first bound the change induced by exchanging on $e$, i.e., $\det(V_{S-v+u}V_{S-v+u}^\top)/\det(V_SV_S^\top)$, and then the change caused by swapping on the remaining edges in $C\backslash e$. The main new ingredient in our proof is the change in the determinant when updating $S-v+u$ by $C\backslash e$ is  approximately the same as the change in determinant when updating $S$ by $C\backslash e$, i.e.,
\begin{align}
    \frac{ \det(V_{S\Delta C}V_{S\Delta C})^\top}{\det(V_SV_S^\top)}  &= \frac{\det(V_{S-v+u}V_{S-v+u}^\top)}{\det(V_SV_S^\top)} \cdot \frac{\det(V_{(S-v+u)\Delta (C\backslash e)}V_{(S-v+u)\Delta (C\backslash e)}^\top)}{\det(V_{S-v+u}V_{S-v+u}^\top)}  \notag \\
    &\approx \frac{\det(V_{S-v+u}V_{S-v+u}^\top)}{\det(V_SV_S^\top)} \cdot \frac{\det(V_{S\Delta(C\backslash e)}V_{S\Delta (C\backslash e)}^\top)}{\det(V_{S}V_{S}^\top)}\,. \label{eq:7}
\end{align}
To prove this, we show that the inner product space induced by  $(V_{S-v+u}V_{S-v+u}^\top)^{-1}$ does not deviate much from the inner product space induced by $(V_SV_S^\top)^{-1}$. As these inner product spaces completely characterize the change in determinant, the inequality follows.

We crucially use the following lemma, first proved in~\cite{BrownLPST22}, to prove quantitative versions of~\eqref{eq:6} and ~\eqref{eq:7}.
\begin{restatable}{lemma}{permanent}  \label{lem:perm}
Let $A_\ell \in \mathbb{R}^{\ell \times \ell}$ satisfy the following properties:
\begin{enumerate}[label=(\roman*)]
    \item \label{th:first}$a_{i,i} > 0$ for all $i \in [\ell]$,
    \item \label{th:second}$0 \leq a_{i,j} \leq  2 \cdot f(i-j)/\prod_{t=j+1}^{i-1} a_{t,t}$ for all $j < i \leq \ell$, and
    \item \label{th:third}$0 \leq a_{i,j} \leq 2 \cdot f(\ell - j + i) \cdot \prod_{t=i}^{j} a_{t,t}/f(\ell)$ for all $i < j \leq \ell$.
\end{enumerate}
Then the permanent of $A_\ell$ is at most $\prod_{i=1}^\ell a_{i,i} \cdot (1 + 0.05/\ell)$. 
\end{restatable} 

\begin{lemma} \label{lem:typeII}
Let $C$ be a minimal $f$-violating cycle in $G(S)$. Then $C$ contains at most one edge of type $\RN{2}$.
\end{lemma}
\begin{proof}\label{proof:lemII}
We prove this lemma by induction on the number of arcs in $C$, $|C|$. Since $G(S)$ is bipartite, $|C|$ is always even. The hypothesis is trivially true for the base case when $|C| = 2$, as $C$ contains only one forward arc. 
Let $|C| = 2\ell$, and for the sake of contradiction, let $C$ contain two type $\RN{2}$ arcs at a distance $k$ from each other, i.e., $C = (u_1 \arc{2} v_1 \rightarrow u_2 \arc{1} v_2 \ldots  u_k \arc{2} v_k \ldots v_\ell \rightarrow u_1)$,
such that $v_i \in S$ and $S - v_i + u_{i \;\text{mod} \; \ell +1}$ is independent in $\mathcal{M}$ for all $i \in [\ell]$. Then consider the two cycles formed by switching the type $\RN{2}$ edges in $C$, 
\begin{align*}
    C_1 = (u_1 \arc{2} v_k \rightarrow u_{k+1}\arc{1} v_{k+2} \ldots v_\ell \rightarrow u_1)\quad \text{and} \quad
    C_2 = (u_k \arc{2} v_1 \rightarrow u_2 \arc{1} v_3   \ldots v_{k-1} \rightarrow u_k)\,.
\end{align*}
Rearranging the weights of the arcs, we get \begin{align*}
    \exp\left(-2w(u_1 \arc{2} v_1) -2w( u_k \arc{2} v_k )\right) &=\frac{1+\normu{1}^2}{1-\normv{1}^2} \cdot  \frac{1+\normu{k}^2}{1-\normv{k}^2} \\
    &= \exp\left(-2w(u_1 \arc{2} v_k) -2w( u_k \arc{2} v_1 )\right)\\
    &=  \frac{\det(V_S V_S^\top + u_1u_1^\top - v_kv_k^\top) }{\det(V_S V_S^\top)} \cdot  \frac{\det(V_S V_S^\top + u_ku_k^\top - v_1v_1^\top) }{\det(V_S V_S^\top)}\,. \tag{from Lemma~\ref{lem:weight_intuition}}
\end{align*}
Since $C$ is $f$-violating, $ \exp(-2w(u_1 \arc{2} v_1) -2w( u_k \arc{2} v_k )) > 0$, and therefore $\det(V_S V_S^\top + u_1u_1^\top - v_kv_k^\top) > 0$ and $\det(V_S V_S^\top + u_ku_k^\top - v_1v_1^\top) > 0$. Consequently, $S+u_1 - v_k$ and $S+u_k - v_1$ are linearly spanning, and $ u_1 \arc{2} v_k $ and $ u_k \arc{2} v_1$ are forward arcs in $G(S)$.

So $G(S)$ contains cycles $C_1$ and $C_2$. Additionally, $w(u_1 \arc{2} v_1) + w( u_k \arc{2} v_k ) = w( u_1 \arc{2} v_k ) + w( u_k \arc{2} v_1)$ implies that $w(C) = w(C_1) + w(C_2)$. Since $\mathrm{ver}(C_1) \subsetneq \mathrm{ver}(C)$, $\mathrm{ver}(C_2) \subsetneq \mathrm{ver}(C)$, and $C$ is a minimal $f$-violating cycle, $C_1$ and $C_2$ are not $f$-violating, i.e., $w(C_1) \geq -\log(f(|C_1|/2)$ and $w(C_2) \geq -\log(f(|C_2|/2)$. Thus
\begin{align*}
    w(C) &= w(C_1) + w(C_2) \geq -\log(f(|C_1|/2)) - \log(f(|C_2|/2)) > -\log(f(|C|/2),
\end{align*}
where the last inequality holds because $f(a) \cdot f(b) < f(a+b)$ and $|C_1|+|C_2| = |C|$. This contradicts the fact that $C$ is $f$-violating, and therefore $C$ can contain at most one type $\RN{2}$ arc.
\end{proof}
We are now ready to prove Theorem~\ref{thm:main2}, which we restate here for convenience.
\main*
\begin{proof}
We defer the proof of $T$ being independent to Lemma~\ref{lem:matroid_indep} in Appendix~\ref{appendix:update} as it is identical to the proof in~\cite{BrownLPST22}. For bounding the determinant of $V_TV_T^\top$, we consider two cases based on the number of type $\RN{2}$ edges in $C$:

\noindent \textbf{Case 1:} $C$ only contains edges of type $\RN{1}$. Let $C = (u_1 \arc{1} v_1 \rightarrow u_2 \arc{1} v_2 \ldots u_\ell \arc{1} v_\ell \rightarrow  u_1)$. Define $X := T\backslash S = \{u_1, \ldots, u_\ell\}$, $Y := S\backslash T = \{v_1, \ldots, v_\ell\}$, and $B := V_X^\top (V_SV_S^\top)^{-1} V_Y$. Using Lemma~\ref{lem:local}, we have \begin{equation}
   \det(V_TV_T^\top) \geq \det(V_S V_S^\top) \cdot \det(V_X^\top (V_S V_S^\top) V_Y)^2 = \det(V_S V_S^\top) \cdot \det(B)^2\,.\label{eq:8}
  \end{equation}
So it suffices to bound $|\det(B)|$ to bound the change in the determinant. We claim that $B$ satisfies the prerequisites of Lemma~\ref{lem:perm}, i.e.,
\begin{itemize}
    \item  for all $i < j \leq \ell$, $|b_{i,j}| \leq 2\cdot f(\ell-j+i)/f(\ell) \cdot \prod_{t=i}^j |b_{t,t}|$ 
    \item for all $\ell \geq i > j$, $|b_{i,j}| \leq 2\cdot f(i-j)/ \prod_{t=j+1}^{i-1} |b_{t,t}|$.
\end{itemize}
For $i, j \in [\ell]$ with $i > j$, define a cycle $C_{i,j} := (u_i \arc{1} v_j \rightarrow u_{j+1}\arc{1} v_{j+1}\ldots  v_{i-1} \rightarrow u_i)$. $C_{i,j}$ contains $2(i-j)$ arcs and $\mathrm{ver}(C_{i,j}) \subsetneq \mathrm{ver}(C)$. So, by minimality of $C$, $C_{i,j}$ is not an $f$-violating cycle. Therefore, $\exp(-w(C_{i,j})) = |\dotp{i}{j}| \cdot \prod_{k = j+1}^{i-1} \exp(-w(u_k \arc{1} v_k)) < f(i-j)$, and
\begin{equation*}
    |b_{i,j}| = |\dotp{i}{j}| < \frac{f(i-j)}{\prod_{k = j+1}^{i-1}\exp(-w(u_k \arc{1} v_k))} = \frac{f(i-j)}{\prod_{k = j+1}^{i-1}|b_{k,k}|}\,,
\end{equation*}
where the last inequality follows from the definition $|b_{i,j}| = |\dotp{i}{j}| = \exp(-w(u_i \arc{1} v_j))$.

For $i,j \in [\ell]$ with $i < j$, define a shortcut cycle $C_{i,j} := (v_\ell \rightarrow u_1 \arc{1} v_1 \ldots u_i \arc{1} v_j\rightarrow v_{j+1} \ldots  u_{\ell} \arc{1} v_{\ell})$.  $C_{i,j}$ contains $2(\ell - j + i) < 2\ell$ arcs and $\mathrm{ver}(C_{i,j}) = \subsetneq \mathrm{ver}(C)$. Again, by the minimality of $C$, $C_{i,j}$ is not an $f$-violating cycle. Therefore, \begin{equation}
  \exp(-w(C_{i,j})) = |\dotp{i}{j}| \cdot \prod_{k = 1}^{i-1} \exp(-w(u_k \arc{1} v_k)) \cdot \prod_{k = j+1}^{\ell} \exp(-w(u_k \arc{1} v_k)) < f(\ell - j + i). \label{eq:9}
\end{equation}
Since $C$ is an $f$-violating cycle, we also have $\prod_{k=1}^{\ell} \exp(-w(u_k \arc{1} v_k)) > f(\ell)$.  Dividing~\eqref{eq:9} by this inequality gives
 \begin{align*}
    |b_{i,j}|= |\dotp{i}{j}| &< \frac{f(\ell-j+i)}{f(\ell)} \cdot \prod_{k=i}^j \exp(-w(u_k \arc{1} v_k)) = \frac{f(\ell-j+i)}{f(\ell)} \cdot \prod_{k=i}^j b_{k,k}\,.
\end{align*}

Therefore applying Lemma~\ref{lem:perm} to $B$, we have $\perm(B) \leq 1.05 \cdot \prod_{i=1}^\ell |b_{i,i}|$. Let $\mathcal{S}_k$ denote the set of permutations of $[k]$, and $\id_k$ denote the identity permutation on $[k]$. Then expanding the determinant of $B$ gives
\begin{align*}
|\det(B)| &= |\sum_{\sigma \in \mathcal{S}_\ell} \prod_{i=1}^\ell b_{i, \sigma(i)}| \geq \prod_{i=1}^\ell |b_{i,i}| - \!\!\sum_{\sigma \in \mathcal{S}_\ell \backslash \id_{\ell}} \prod_{i=1}^\ell| b_{i, \sigma(i)} | \\
&{}\geq{} \prod_{i=1}^\ell |b_{i,i}| - (\perm(|B|)-\prod_{i=1}^\ell |b_{i,i}|)\\
&{}\geq{} 0.95 \cdot \prod_{i=1}^\ell |b_{i,i}| \,.
\end{align*}

Since $C$ is an $f$-violating, the exponential of $w(C)$ is equal to \[\prod_{i=1}^\ell \exp(-w(u_i \arc{1} v_i)) = \prod_{i=1}^\ell |b_{i,i}| \geq f(\ell),\]
and therefore $|\det(B)| \geq 0.95 f(\ell)$. Plugging this bound in equation~\eqref{eq:8} gives \begin{equation*}
   \det(V_TV_T^\top) \geq \det(V_S V_S^\top) \cdot |\det(B)|^2 \geq (0.95\cdot f(\ell))^2 > 2\,. 
\end{equation*}

\noindent \textbf{Case 2}: Cycle $C$ contains exactly one edge of type $\RN{2}$. Let $ C = (u_1 \arc{1} v_1 \rightarrow u_2 \arc{1} v_2 \ldots u_\ell \arc{2} v_\ell \rightarrow  u_1)$, and define $X := T \backslash S = \{u_1, \ldots, u_\ell\}$, $Y := S\backslash T = \{v_1, \ldots, v_\ell\}$.

Since there could be two forward arcs from a vertex $u_i$ to a vertex $v_j$, we define a matrix $A \in \mathbb{R}^{\ell \times \ell}$ such that the $(i,j)$-th entry of $A$ corresponds to the arc with the lower weight. Let $[A]_{i,j} := a_{i,j}$ such that  
\begin{itemize}
    \item $a_{i,j} = \max\{|\langle u_i, v_j\rangle_S|, \sqrt{(1+\norm{u_i}_S^2)\cdot (1-\norm{v_j}_S^2)}\}$ for all $i,j < \ell$,
    \item $a_{\ell, j} = \sqrt{(1+\norm{u_\ell}_S^2) \cdot (1-\norm{v_j}_S^2) + \langle u_\ell, v_j\rangle_S^2}$ for all $j < \ell$, and
    \item $a_{i, \ell} = \sqrt{(1+\norm{u_i}_S^2) \cdot (1-\norm{v_\ell}_S^2) + \langle u_i, v_\ell\rangle_S^2}$ for all $i$.
\end{itemize}

In the rest of the proof, we show that $\det(V_TV_T^\top) \geq  0.65 \cdot (\prod_{i=1}^\ell a_{i,i})^2 \cdot \det(V_S V_S^\top)$. This proves the required increase in the determinant because of the following claim.
\begin{claim} The entries of the matrix $A$ satisfy\label{claim:f_viol}
\begin{enumerate}[label=(\alph*)]
    \item \label{item:a}$\prod_{i=1}^\ell a_{i,i} \geq f(\ell)$,
    \item \label{item:b}$a_{i,i} = |\dotp{i}{i}|$ for all $i \in [\ell-1]$, and
    \item \label{item:c}for any $k \in [\ell]$, the principal submatrix of $A$ corresponding to the first $k$ rows and columns, $A_{k,k}$ satisfies $\perm(A_{k,k}) \leq (1 + 0.05/k) \cdot \prod_{i=1}^k a_{i,i} $.
\end{enumerate}
\end{claim}
The proof of the claim appears in Appendix~\ref{sec:omitted_proofs}.

We first bound the change in the determinant caused by swapping on the arc $(u_\ell \arc{2} v_\ell)$. Let $S_1 = S - v_\ell + u_\ell$. Using Lemma~\ref{lem:det_update} on $S_1$ gives
\begin{align}
    \det(V_{S_1}V_{S_1}^\top) &=  \det(V_S V_S^\top) \cdot \det\left(I_2 + \begin{bmatrix}u_\ell^\top (V_S V_S^\top)^{-1} u_\ell & u_\ell^\top (V_S V_S^\top)^{-1}  v_\ell \\
    -v_\ell^\top (V_S V_S^\top)^{-1}  u_\ell & -v_\ell^\top (V_S V_S^\top)^{-1}  v_\ell\end{bmatrix}\right) \notag \\
    &= \det(V_S V_S^\top) \cdot \left((1 + \normu{\ell}^2)\cdot (1 - \normv{\ell}^2) + \dotp{\ell}{\ell}^2\right) =  \det(V_S V_S^\top) \cdot a_{\ell, \ell}^2 \,. \label{eq:10}
\end{align}
Note that $(1 + \normu{\ell}^2)\cdot (1 - \normv{\ell}^2) = \exp(-2\cdot w(u_\ell \arc{2} v_\ell)) > 0$ since the edge $(u_\ell \arc{2} v_\ell)$ belongs to the $f$-violating cycle $C$. This fact, along with $ \det(V_S V_S^\top) > 0$, implies that $\det(V_{S_1}V_{S_1}^\top) > 0$ and, therefore, $V_{S_1}V_{S_1}^\top$ is invertible.

Let $X_1 = X -u_\ell$ and $Y_1 = Y - v_\ell$. Then by Lemma~\ref{lem:local},
\begin{align}
    \det(V_TV_T^\top) &= \det(V_{S_1}V_{S_1}^\top - V_{Y_1}V_{Y_1}^\top + V_{X_1}V_{X_1}^\top) \geq \det(V_{S_1}V_{S_1}^\top)\cdot  \det(V_{X_1}^\top (V_{S_1}V_{S_1}^\top)^{-1} V_{Y_1})^2\,. \label{eq:11}
\end{align}

Combining~\eqref{eq:10} and~\eqref{eq:11}, we get
\begin{equation}
   \det(V_TV_T^\top) \geq \det(V_S V_S^\top)\cdot a_{\ell,\ell}^2 \cdot  \det(V_{X_1}^\top (V_{S_1}V_{S_1}^\top)^{-1} V_{Y_1})^2\,. \label{eq:12}
\end{equation}
In the rest of the proof, we show that $| \det(V_{X_1}^\top (V_{S_1}V_{S_1}^\top)^{-1} V_{Y_1})| \approx | \det(V_{X_1}^\top (V_{S_1}V_{S_1}^\top)^{-1} V_{Y_1})|$. We start by measuring the deviation of the inner product space induced by $(V_{S_1}V_{S_1}^\top)^{-1}$ from the inner product space induced by $(V_SV_S^\top)^{-1}$, and then use the minimality of cycle $C$ to control this deviation.

Since $(V_{S_1}V_{S_1}^\top)^{-1}  = (V_S V_S^\top - u_\ell u_\ell^\top + v_\ell v_\ell^\top )^{-1} = \left(V_S V_S^\top + \begin{bmatrix} u_\ell & -v_\ell \end{bmatrix} \cdot \begin{bmatrix} u_\ell^\top  \\ v_\ell^\top \end{bmatrix} \right)^{-1}$, applying the Woodbury matrix identity to $(V_{S_1}V_{S_1}^\top)^{-1}$ gives
\begin{align}
  (V_{S_1}V_{S_1}^\top)^{-1} &- (V_S V_S^\top)^{-1} \notag \\
    &=  - (V_S V_S^\top)^{-1}\begin{bmatrix} u_\ell & -v_\ell \end{bmatrix}  \cdot \left(I_2 + \begin{bmatrix} u_\ell^\top \\ v_\ell^\top \end{bmatrix} (V_S V_S^\top)^{-1} \begin{bmatrix} u_\ell & -v_\ell \end{bmatrix}\right)^{-1} \cdot \begin{bmatrix} u_\ell^\top \\ v_\ell^\top \end{bmatrix} (V_S V_S^\top)^{-1} \notag\\
    &= - (V_S V_S^\top)^{-1}\begin{bmatrix} u_\ell & -v_\ell \end{bmatrix}  \cdot \begin{bmatrix} 1+\normu{\ell}^2 & -\dotp{\ell}{\ell} \\
    \dotp{\ell}{\ell} & 1-\normv{\ell}^2
  \end{bmatrix}^{-1} \cdot \begin{bmatrix} u_\ell^\top \\ v_\ell^\top \end{bmatrix} (V_S V_S^\top)^{-1}\,. \label{eq:13}
\end{align}
For any invertible matrix $W$, $W^{-1} = \mathrm{adjoint}(W)/\det(W)$. Therefore,
\begin{align*}
    \begin{bmatrix} 1+\normu{\ell}^2 & -\dotp{\ell}{\ell} \\
    \dotp{\ell}{\ell} & 1-\normv{\ell}^2
  \end{bmatrix}^{-1} &=  \frac{1}{(1+\normu{\ell}^2)(1-\normv{\ell}^2)+\dotp{\ell}{\ell}^2} \cdot \begin{bmatrix} 1-\normv{\ell}^2 & \dotp{\ell}{\ell} \\
    -\dotp{\ell}{\ell} & 1 + \normu{\ell}^2\end{bmatrix}\\
    &=  a_{\ell, \ell}^2 \cdot \begin{bmatrix} 1-\normv{\ell}^2 & \dotp{\ell}{\ell} \\
    -\dotp{\ell}{\ell} & 1 + \normu{\ell}^2\end{bmatrix}. \tag{from definition of $a_{\ell, \ell}$}
\end{align*}
Substituting the inverse in equation~\ref{eq:13} gives
\begin{align*}
    (V_{S_1}V_{S_1}^\top)^{-1} &- (V_S V_S^\top)^{-1} \\
    &=  - a_{\ell,\ell}^{-2}\cdot(V_S V_S^\top)^{-1}\begin{bmatrix} u_\ell & -v_\ell \end{bmatrix} \cdot \begin{bmatrix} 1-\normv{\ell}^2 & \dotp{\ell}{\ell} \\
    -\dotp{\ell}{\ell} & 1 + \normu{\ell}^2\end{bmatrix} \cdot \begin{bmatrix} u_\ell^\top \\ v_\ell^\top \end{bmatrix} (V_S V_S^\top)^{-1} \\
    &= - a_{\ell,\ell}^{-2}\cdot \Big[(V_S V_S^\top)^{-1}u_\ell u_\ell^\top (V_S V_S^\top)^{-1} (1-\normv{\ell}^2)+(V_S V_S^\top)^{-1}v_\ell u_\ell^\top (V_S V_S^\top)^{-1}\dotp{\ell}{\ell} \\
    &\quad +(V_S V_S^\top)^{-1}u_\ell v_\ell^\top (V_S V_S^\top)^{-1}\dotp{\ell}{\ell}-(V_S V_S^\top)^{-1}v_\ell v_\ell^\top (V_S V_S^\top)^{-1}(1 + \normu{\ell}^2)  \Big]\,.
\end{align*}
So for $i,j \in [\ell-1]$, the $i,j$-th entry of $V_{X_1}^\top (V_{S_1}V_{S_1}^\top)^{-1} V_{Y_1}$ is given by
\begin{align}
     u_i^\top (V_{S_1}V_{S_1}^\top)^{-1} v_j  &=\dotp{i}{j} - 
    a_{\ell,\ell}^{-2}\cdot \Big[(1-\normv{\ell}^2)\cdot \dotu{i}{\ell}\cdot \dotp{\ell}{j}+\dotp{\ell}{\ell}\cdot \dotp{i}{\ell} \cdot \dotp{\ell}{j}   \notag \\
    &\quad +\dotp{\ell}{\ell} \cdot \dotu{i}{\ell}\cdot \dotv{\ell}{j}-(1+\normu{\ell}^2) \cdot \dotp{i}{\ell} \cdot \dotv{\ell}{j} \Big]\,. 
\label{eq:14}
\end{align}

The difference between $ u_i^\top (V_{S_1}V_{S_1}^\top)^{-1} v_j $ and $  u_i^\top (V_SV_S^\top)^{-1} v_j $ consists of four terms bilinear in $u_i$ and $v_j$. Even though these terms do not directly correspond to any arc weights in $G(S)$, each of these bilinear terms can be upper bounded by entries of the matrix $A$ (see in Claim~\ref{claim:z}). We then use the multi-linearity of the determinant to show that their contribution to $\det(V_{X_1}(V_{S_1}V_{S_1}^\top)^{-1}V_{Y_1} )$ is negligible, and therefore $|\det(V_{X_1}(V_{S_1}V_{S_1}^\top)^{-1}V_{Y_1} )|$ is approximately equal to $|\det(A )|$.

By equation~\eqref{eq:14}, the $j$-th column of $V_{X_1}(V_{S_1}V_{S_1}^\top)^{-1} V_{Y_1}$ is the sum of the five vectors defined below:
\begin{align*}
    z_j^0 &:= V_{X_1}^\top (V_S V_S^\top)^{-1} v_j, \\
    z_j^1 &:= -a_{\ell,\ell}^{-2}\cdot(1-\normv{\ell}^2) \cdot (V_{X_1}^\top (V_S V_S^\top)^{-1} u_\ell) \cdot \dotp{\ell}{j},  \\
    z_j^2 &:= -a_{\ell,\ell}^{-2}\cdot\dotp{\ell}{\ell} \cdot (V_{X_1}^\top (V_S V_S^\top)^{-1} u_\ell) \cdot \dotv{\ell}{j}\,,  \\
    z_j^3 &:= -a_{\ell,\ell}^{-2}\cdot\dotp{\ell}{\ell} \cdot ( V_{X_1}^\top (V_S V_S^\top)^{-1} v_\ell) \cdot \dotp{\ell}{j},\\
    z_j^4 &:= a_{\ell,\ell}^{-2}\cdot (1+\normu{\ell}^2) \cdot (V_{X_1}^\top (V_S V_S^\top)^{-1} v_\ell) \cdot \dotv{\ell}{j}\,.
\end{align*}

So the $j$-th column of $V_{X_1}^\top (V_S V_S^\top)^{-1} V_{Y_1}$ equals $ z_j^0 +  z_j^1 +  z_j^2+  z_j^3+  z_j^4$.

To bound the deviation terms, we crucially use the following claim about the absolute value of vector $z_j^k$. A proof of the claim appears in Appendix~\ref{appendix:update}.
\begin{claim}\label{claim:z}
The matrix $A$ and the vectors $z_j^0, z_j^1, z_j^2,z_j^3,z_j^4$ satisfy
\begin{itemize}
    \item for $k \in \{1,2,3,4\}$ and $i \in [\ell-1]$, $|z_j^k(i)| \leq a_{i, \ell} \cdot a_{\ell, j} \cdot a_{\ell, \ell}^{-1}$,

    \item for any $j_1 \neq j_2$, and $i_1 \neq i_2$,  $ |z_{j_1}^1(i_1) \cdot z_{j_2}^4(i_2)| \leq a_{i_1, j_2}  \cdot a_{\ell, j_1}\cdot a_{i_2, \ell}\cdot a_{\ell, \ell}^{-1}$, and  $|z_{j_1}^2(i_1) \cdot z_{j_2}^3(i_2)| \leq a_{i_1, \ell}  \cdot a_{\ell, j_2}\cdot a_{i_2, j_1} \cdot a_{\ell, \ell}^{-1}$\,.
\end{itemize}
\end{claim}
Using this claim, we upper bound entries from $z^k_j$ by entries of matrix $A$ later in the proof.

Let $\det(w_1, w_2, \ldots, w_{\ell})$ denote the determinant of the matrix whose $i$-th column is $w_i$. Using the multi-linearity of the determinant, \begin{align}
    \det(V_{X_1}^\top (V_{S_1}V_{S_1}^\top)^{-1} V_{Y_1}) = \det(\sum_{k=0}^4 z_1^{k}, \sum_{k=0}^4 z_2^{k}, \ldots, \sum_{k=0}^4 z_{\ell-1}^{k})= \sum_{q \in \mathcal{Q}} \det(z_1^{q(1)}, z_2^{q(2)}, \ldots, z_{\ell-1}^{q(\ell-1)})\, , \label{eq:15}
\end{align}
where $\mathcal{Q}$ is the set of all functions from $[\ell-1]$ to $\{0,1,2,3,4\}$, i.e., $\mathcal{Q} = \{q:[\ell-1] \rightarrow \{0,1,2,3,4\}\}$. Note that for any $j_1, j_2 \in [\ell-1]$ and $i \in \{1,2,3,4\}$, 
the vectors $z_{j_1}^i$ and $z_{j_2}^i$ are co-linear. Additionally, for any $j \in [\ell-1]$,
$z_{j}^1$ and $z_{j}^2$ are both co-linear with $V_{X_1}^\top (V_S V_S^\top)^{-1} u_\ell$, which does not depend on $j$. Similarly, $z_{j}^3$ and $z_{j}^4$ are both co-linear with $V_{X_1}^\top (V_S V_S^\top)^{-1} v_\ell$. So for any function $q \in \mathcal{Q}$, if more than two columns $j$ have $q(j) \neq 0$, then $\det(z_1^{q(1)}, z_2^{q(2)}, \ldots, z_{\ell-1}^{q(\ell-1)}) = 0$.

Let $z^k_{j}(i)$ denote the $i$-th entry of vector $z^k_j$. We can can alter the formula for $z^k_j$ to obtain a formula for $z^k_{j}(i)$ by replacing the $(V_{X_1}^\top(V_SV_S^\top)^{-1}v_\ell)$ term by its $i$-th entry, which is exactly $u_{i}^\top(V_SV_S^\top)^{-1} v_\ell = \dotp{i}{\ell}$.

To further restrict the class of functions which contribute to $\det(V_{X_1} (V_{S_1}V_{S_1}^\top)^{-1} V_{Y_1})$, note that
\begin{align}
    z^1_{j_1}(i_1) \cdot z^3_{j_2}(i_2) &= -a_{\ell,\ell}^{-4}\cdot (1-\normv{\ell}^2) \cdot \dotp{\ell}{\ell} \cdot \dotp{\ell}{j_1} \cdot \dotp{\ell}{j_2} \cdot \dotu{i_1}{\ell} \cdot \dotp{i_2}{\ell}\notag \\
    &=z^3_{j_1}(i_2) \cdot z^1_{j_2}(i_1)\,. \label{eq:16} 
\end{align}
Now consider two columns $j_1, j_2 \in [\ell-1]$ with $j_1 \neq j_2$  and functions $q_1, q_2 \in \mathcal{Q}$ such that $q_1(j) = q_2(j) = 0$ for all $j \notin \{j_1, j_2\}$, $q_1(j_1) = 1, q_1(j_2) = 3$, and $q_2(j_1) = 3, q_2(j_2) = 1$. Since exchanging two columns of a matrix flips the sign of the determinant, equation~\eqref{eq:16} implies that
\begin{equation*}
    \det(z_1^{q_1(1)}, z_2^{q_1(2)},\ldots, z_{\ell-1}^{q_1(\ell-1)}) = -  \det(z_1^{q_2(1)}, z_2^{q_2(2)}, \ldots, z_{\ell-1}^{q_2(\ell-1)})\,.
\end{equation*}
The terms corresponding to $q_1$ and $q_2$ cancel each other, so we can ignore their contribution to the overall sum in~\eqref{eq:15}.

Similarly, if  $q_1(j_1) = 2, q_1(j_2) = 4$ and $q_2(j_1) = 4, q_2(j_2) = 2$, then \begin{equation*}
    \det(z_1^{q_1(1)}, z_2^{q_1(2)},\ldots, z_{\ell-1}^{q_1(\ell-1)}) = -\det(z_1^{q_2(1)}, z_2^{q_2(2)}, \ldots, z_{\ell-1}^{q_2(\ell-1)}).
\end{equation*}

So the subset of function in $\mathcal{Q}$ with non-zero determinants is the union of the following four sets: \begin{itemize}
    \item $\{q_0\}$ such that $q_0(j) = 0$ for all $j \in [\ell-1]$,
    \item $\mathcal{Q}_2 = \{q \in \mathcal{Q}: \,q(j) \neq 0 \text{ for exactly one } j\}$,
    \item $\mathcal{Q}_3 = \{f \in \mathcal{Q}: \,\exists j_1, j_2 \text{ s.t. }q(j_1) = 1$, $q(j_2) = 4, \text{ and } q(j) = 0 \; \forall  j \in [\ell-1] \backslash \{j_1, j_2\} \}$,
     \item $\mathcal{Q}_4 = \{f \in \mathcal{Q}: \,\exists j_1, j_2 \text{ s.t. }q(j_1) = 2$, $q(j_2) = 3, \text{ and } q(j) = 0 \; \forall  j \in [\ell-1] \backslash \{j_1, j_2\} \}$.
\end{itemize}

Using these definitions and the triangle inequality, equation~\eqref{eq:15} implies
\begin{equation}
    |\det(V_{X_1}^\top (V_{S_1}V_{S_1}^\top)^{-1} V_{Y_1})| \geq |\det(z_1^{0}, z_2^{0}, \ldots, z_{\ell-1}^{0})| - \sum_{q \in \{\mathcal{Q}_2 \cup \mathcal{Q}_3 \cup \mathcal{Q}_4\} } |\det(z_1^{q(1)}, z_2^{q(2)}, \ldots, z_{\ell-1}^{q(\ell-1)})|. \label{eq:17}
\end{equation}

By expanding the determinants, we now bound every determinant term in~\eqref{eq:17} separately. Let $\mathcal{S}_\ell$ denote the set of permutations on $[\ell]$ and $\id_\ell$ denote the identity permutation on $\ell$ elements. Expanding the first term in~\eqref{eq:17} gives
\begin{align*}
   |\det(z_1^{0}, z_2^{0}, \ldots, z_{\ell-1}^{0})| &= |\sum_{\sigma \in \mathcal{S}_{\ell-1}} \prod_{i=1}^{\ell-1} z_i^0(\sigma(i))| \geq \prod_{i=1}^{\ell-1} |z^0_{i}(i)| - \sum_{\sigma \in \mathcal{S}_{\ell-1} \backslash \id_{\ell-1}} \prod_{j} |z^0_j(\sigma(j))|,
\end{align*}
where the last inequality follows from the triangle inequality.
By part~\ref{item:b} of Claim~\ref{claim:f_viol}, $|z_i^0(i)| = |\dotp{i}{i}| = a_{i,i}$ for all $i < \ell$ and by definition, $|z_j^0(i)| = |\dotp{i}{j}| \leq a_{i,j}$ for $i\neq j$ and $i < \ell$. Therefore,
\begin{align*}
    |\det(z_1^{0}, z_2^{0}, \ldots, z_{\ell-1}^{0})| &\geq \prod_{i=1}^{\ell-1} a_{i, i} - \sum_{\sigma \in \mathcal{S}_{\ell-1} \backslash \id_{\ell-1}} \prod_{j=1}^{\ell-1} a_{\sigma(j), j} = \prod_{i=1}^{\ell-1} a_{i,i} - (\perm(A_{\ell-1, \ell-1}) - \prod_{i=1}^{\ell-1} a_{i,i})\,,
\end{align*}
where $A_{\ell-1, \ell-1}$ is the principal submatrix of $A$ consisting of the first $\ell-1$ rows and columns. Using part~\ref{item:c} of Claim~\ref{claim:f_viol}, we have $\perm(A_{\ell-1}) \leq 1+0.05/(\ell-1)\cdot \prod_{i=1}^{\ell-1} a_{i,i}$, and therefore
\begin{equation}
     |\det(z_1^{0}, z_2^{0}, \ldots, z_{\ell-1}^{0})|  \geq 0.95\cdot \prod_{i=1}^{\ell-1} a_{i,i}\,. \label{eq:18}
\end{equation}

Summing up the determinants corresponding to the functions in $\mathcal{Q}_2$, we have
\begin{align}
    \sum_{f \in \mathcal{Q}_2}& |\det(z_1^{q(1)}, z_2^{q(2)}, \ldots, z_{\ell-1}^{q(\ell-1)})| = \sum_{j \in [\ell-1]}\sum_{i \in [4]}  |\det(z_1^{0}, z_2^{0}, \ldots, z_{j}^{i}, \ldots, z_{\ell-1}^{0})| \notag\\
    &= \sum_{j \in [\ell-1]}\sum_{i \in [4]}  |\sum_{\sigma \in \mathcal{S}_{\ell-1}} \left(\prod_{x\neq j} z_{x}^{0}(\sigma(x)) \cdot z_{j}^{i}(\sigma(j))\right)| \notag\\
    &\leq\sum_{j \in [\ell-1]}\sum_{i \in [4]} \sum_{\sigma\in \mathcal{S}_{\ell-1}} \left(\prod_{x\neq j} |z_{x}^{0}(\sigma(x))|\right) \cdot |z_{j}^{i}(\sigma(j))|\notag\\
    &\leq \sum_{j \in [\ell-1]}\sum_{i \in [4]} \sum_{\sigma\in \mathcal{S}_{\ell-1}} \left(\prod_{x\neq j} a_{\sigma(x), x}\right) \cdot a_{\sigma(j), \ell}\cdot a_{\ell, j}\cdot a_{\ell, \ell}^{-1} \tag*{(from Claim~\ref{claim:z})} \notag\\
     &\leq 4\cdot\sum_{j \in [\ell-1]} \sum_{\sigma\in \mathcal{S}_{\ell-1}} \left(\prod_{x\neq j} a_{\sigma(x), x}\right) \cdot a_{\sigma(j), \ell}\cdot a_{\ell, j}\cdot a_{\ell, \ell}^{-1}\,. \label{eq:19}
\end{align}     

To simplify the last expression, consider a function $\mathcal{T}: S_{\ell-1} \times [\ell-1] \rightarrow \mathcal{S}_{\ell}$ mapping every (permutation, index) pair $(\sigma \in S_{\ell-1}, j \in [\ell-1])$ to a permutation in $\mathcal{S}_\ell$. Define $\mathcal{T}(\sigma, j)$ as the permutation which maps $\mathcal{T}(\sigma, j)(j) = \ell$, $\mathcal{T}(\sigma, j)(\ell) = \sigma(j)$, and $\mathcal{T}(\sigma, j)(x) = \sigma(x)$ for all $x \in [\ell-1]\backslash \{j\}$.  Using this definition,  $ \left(\prod_{x\neq j} a_{\sigma(x), x} \right)\cdot a_{\sigma(j), \ell}\cdot a_{\ell, j} = \prod_{x=1}^\ell a_{\mathcal{T}(\sigma, j)(x), x}$\, (see Figure~\ref{fig:absolute_z}). 
\begin{figure}[h]
    \centering
    \includegraphics[width=0.7\textwidth]{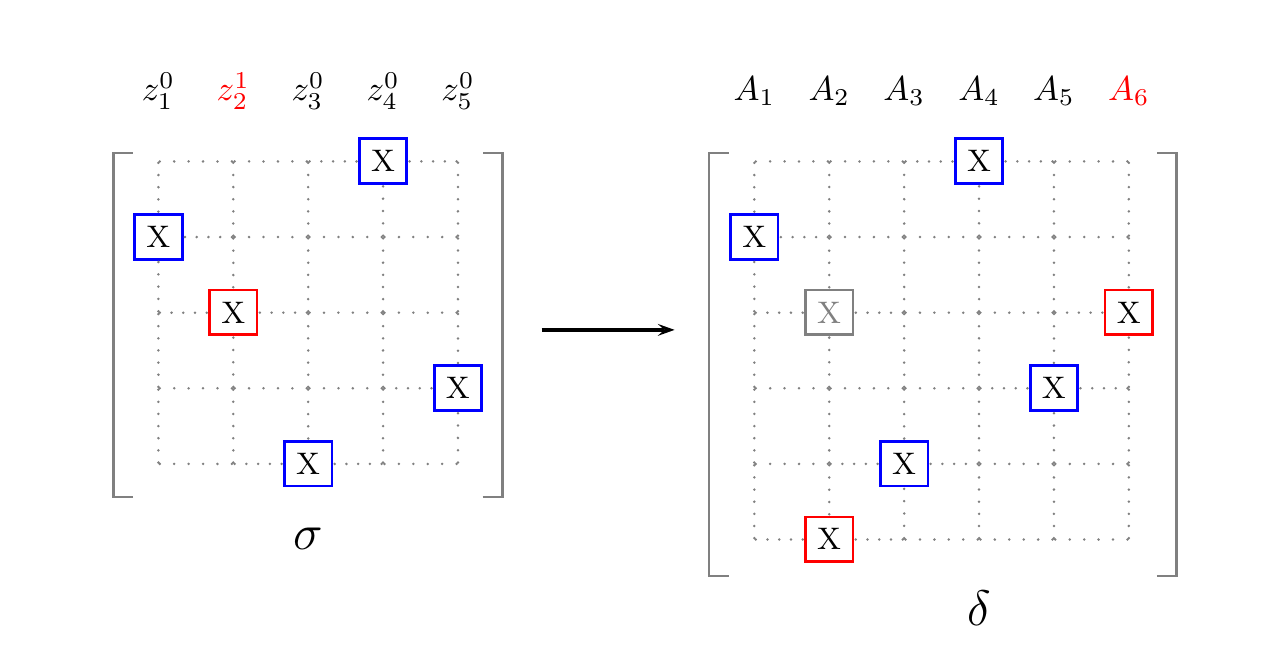}
    \caption{\small In expansion of $\det(z_1^0, z_2^1, z_3^0, z_4^0, z_5^0)$, the product of entries for permutation $\sigma \in \mathcal{S}_5$ with $\sigma(2) = 3$ and $j=2$ can be upper bounded by a permutation $\delta \in \mathcal{S}_{6}$ in expansion of $\perm(A)$ with $\delta(2) = 6, \delta(6) = 3$, and all other entries same as $\sigma$.}
    \label{fig:absolute_z}
\end{figure}

Since $\mathcal{T}$ is one-to-one by construction and does not include $\id_{\ell}$ in its image, the summation can be further simplified as
\begin{align}
   \sum_{j \in [\ell-1]} \sum_{\sigma\in \mathcal{S}_{\ell-1}} \left(\prod_{x\neq j} a_{\sigma(x), x} \right)\cdot a_{\sigma(j), \ell} =  \sum_{j \in [\ell-1]} \sum_{\sigma\in \mathcal{S}_{\ell-1}} \left(\prod_{x=1}^{\ell} a_{\mathcal{T}(\sigma, j)(x), x}\right) = \sum_{\delta\in \mathcal{S}_{\ell} \backslash \id_\ell} \prod_{x=1}^{\ell} a_{\delta(x), x}\,. \label{eq:20}
\end{align}

Plugging~\eqref{eq:20} in~\eqref{eq:19}, we get
\begin{align}
    \sum_{q \in \mathcal{Q}_2} |\det(z_1^{q(1)}, z_2^{q(2)}, \ldots, z_{\ell-1}^{q(\ell-1)})| &\leq 4  \cdot \sum_{\delta\in \mathcal{S}_{\ell} \backslash \id_{\ell}} \prod_{x=1}^{\ell} a_{\delta(x), x}  \cdot a_{\ell, \ell}^{-1}\notag \\
    &\leq 4 \cdot (\perm(A) - \prod_{i=1}^\ell a_{i,i})\cdot a_{\ell, \ell}^{-1}  \leq  0.2 \cdot \prod_{i=1}^{\ell-1} a_{i,i}\,, \label{eq:21}
\end{align}
where the last inequality follows by part~\ref{item:a} of Claim~\ref{claim:f_viol}.

In a similar fashion, summing up the determinants corresponding to the functions in $\mathcal{Q}_3$, we have
\begingroup
\allowdisplaybreaks
\begin{align}
    \sum_{q \in \mathcal{Q}_3} &|\det(z_1^{q(1)}, z_2^{q(2)}, \ldots, z_{\ell-1}^{q(\ell-1)})| 
    = \sum_{\substack{j_1, j_2 \in [\ell-1]\\ q(j_1) = 1, q(j_2) = 4}}  |\det(z_1^{q(1)}, z_2^{q(2)}, \ldots, z_{\ell-1}^{q(\ell-1)})|\notag\\
    &= \sum_{\substack{j_1, j_2 \in [\ell-1]\\ q(j_1) = 1, q(j_2) = 4}}  \left|\sum_{\sigma \in \mathcal{S}_{\ell-1}} \prod_{j=1}^{\ell-1} z_j^{q(j)}(\sigma(j))\right| \leq \sum_{\substack{j_1, j_2 \in [\ell-1]\\ q(j_1) = 1, q(j_2) = 4}}  \sum_{\sigma \in \mathcal{S}_{\ell-1}} \prod_{j=1}^{\ell-1} |z_j^{q(j)}(\sigma(j))|  \notag\\
    &= \sum_{\substack{j_1, j_2 \in [\ell-1]\\ q(j_1) = 1, q(j_2) = 4}}  \sum_{\sigma \in \mathcal{S}_{\ell-1}} \left(\prod_{j \neq j_1, j_2} |z_j^{0}(\sigma(j))| \right) \cdot |z_{j_1}^{1}(\sigma(j_1)) \cdot z_{j_2}^{4}(\sigma(j_2))| \notag\\
    &\leq \sum_{\substack{j_1, j_2 \in [\ell-1]\\ q(j_1) = 1, q(j_2) = 4}}  \sum_{\sigma \in \mathcal{S}_{\ell-1}} \left(\prod_{j \neq j_1, j_2} a_{\sigma(j), j} \right)\cdot |z_{j_1}^{1}(\sigma(j_1)) \cdot z_{j_2}^{4}(\sigma(j_2))|\,. \tag*{(from definition of $z^0$)} \notag\\
    &\leq \sum_{\substack{j_1, j_2 \in [\ell-1]\\ q(j_1) = 1, q(j_2) = 4}}  \sum_{\sigma \in \mathcal{S}_{\ell-1}} \left(\prod_{j \neq j_1, j_2} a_{\sigma(j), j} \right) \cdot a_{\sigma(j_1), j_2}  \cdot a_{\ell, j_1}\cdot a_{\sigma(j_2), \ell} \cdot a_{\ell, \ell}^{-1}\,,\label{eq:22}
\end{align}
\endgroup
where the last inequality follows from Claim~\ref{claim:z}.
Again, consider a function $\mathcal{T}: \mathcal{S}_{\ell-1} \times [\ell-1] \times [\ell-1] \rightarrow \mathcal{S}_\ell$ that maps every (permutation, index, index) tuple $(\sigma \in S_{\ell-1}, j_1\in [\ell-1], j_2 \in [\ell-1])$ with $j_1 \neq j_2$ to a permutation $\mathcal{S}_\ell$. Define $\mathcal{T}(\sigma, j_1, j_2)$ as the permutation in $\mathcal{S}_\ell$ with $\mathcal{T}(\sigma, j_1, j_2)(j_1) = \ell$, $\mathcal{T}(\sigma, j_1, j_2)(j_2) = \sigma(j_1)$, $\mathcal{T}(\sigma, j_1, j_2)(\ell) = \sigma(j_2)$, and $\mathcal{T}(\sigma, j_1, j_2)(x) = \sigma(x)$ for all $x \neq \{j_1, j_2\}$. Using this mapping, $ \left(\prod_{j \neq j_1, j_2} a_{\sigma(j), j} \right) \cdot a_{\sigma(j_1), j_2}  \cdot a_{\ell, j_1}\cdot a_{\sigma(j_2), \ell} = \prod_{x=1}^\ell a_{\mathcal{T}(\sigma, j_1, j_2)(x), x}$\, (see Figure~\ref{fig:absolute_z1z2}).
\begin{figure}[h]
    \centering
    \includegraphics[width=0.7\textwidth]{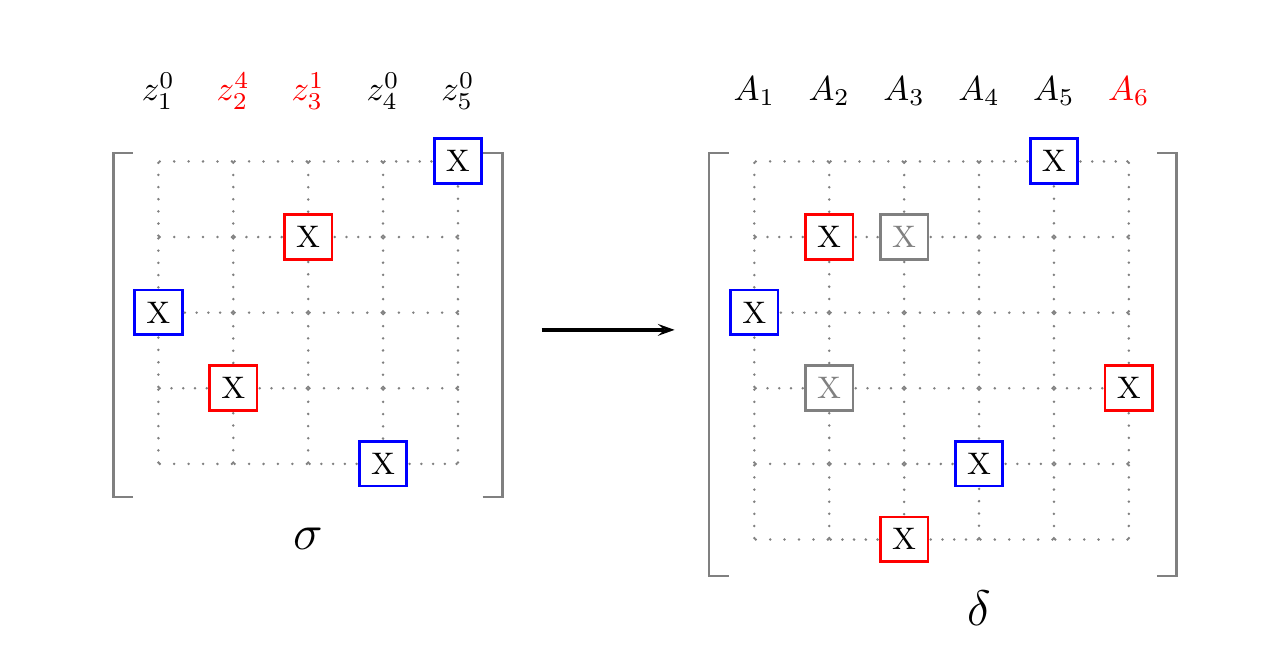}
    \caption{\small In expansion of $\det(z_1^0, z_2^4, z_3^1, z_4^0, z_5^0)$, the term corresponding to permutation $\sigma \in \mathcal{S}_5$ with $\sigma(3) = 2$, $\sigma(2) = 4$ with $j_1 = 3$ and $j_2 = 2$,  can be upper bounded by a permutation $\delta \in \mathcal{S}_{6}$ in expansion of $\perm(A)$ with $\delta(2) = 2, \delta(3) = \ell, \delta(\ell) = 4$, and all other entries same as $\sigma$.}
    \label{fig:absolute_z1z2}
\end{figure}

Note that for a fixed input $j_2$, $\mathcal{T}$ is a one-to-one function in $j_1$ and $\sigma$. So given a permutation $\delta \in \mathcal{S}_\ell$, there are exactly $\ell$ tuples $(\sigma, j_1, j_2)$ such that $\mathcal{T}(\sigma, j_1, j_2) = \delta$. Also, since $\sigma(j_2) < \ell$, the image of $\mathcal{T}$ does not contain $\id_\ell$. So summing over all terms in $\mathcal{Q}_3$ gives
\begin{align*}
     \sum_{\substack{j_1, j_2 \in [\ell-1]\\ q(j_1) = 1, q(j_2) = 4}}  \sum_{\sigma \in \mathcal{S}_{\ell-1}} \left(\prod_{j \neq j_1, j_2} a_{\sigma(j), j} \right) \cdot a_{\sigma(j_1), j_2}  \cdot a_{\ell, j_1}\cdot a_{\sigma(j_2), \ell} &= 
    \sum_{\substack{j_1, j_2 \in [\ell-1]\\ q(j_1) = 1, q(j_2) = 4}}  \sum_{\sigma \in \mathcal{S}_{\ell-1}}\prod_{x=1}^\ell a_{\mathcal{T}(\sigma, j_1, j_2)(x), x}\\
    &= \ell \cdot \sum_{\delta \in \mathcal{S}_{\ell} \backslash \id_\ell}\prod_{x=1}^\ell a_{\delta(x), x}\,.
\end{align*}

Plugging this into~\eqref{eq:22} gives
\begin{align}
     \sum_{q \in \mathcal{Q}_3} &|\det(z_1^{q(1)}, z_2^{q(2)}, \ldots, z_{\ell-1}^{f(\ell-1)})| \leq     \ell\cdot \sum_{\delta \in \mathcal{S}_{\ell} \backslash \id_\ell} \left(\prod_{x=1}^\ell a_{\delta(x), x} \right)  \cdot a_{\ell, \ell}^{-1}\, \notag \\
     &=\ell \cdot(\perm(A)-\prod_{i=1}^\ell a_{i,i})  \cdot a_{\ell, \ell}^{-1} \leq 0.05 \cdot \prod_{i=1}^{\ell-1} a_{i,i}\,, \label{eq:23}
\end{align}
where the last inequality follows from Claim~\ref{claim:f_viol}.

Similarly, summing over the functions in $\mathcal{Q}_4$ gives
\begin{align}
    \sum_{q \in \mathcal{Q}_4} &|\det(z_1^{q(1)}, z_2^{q(2)}, \ldots, z_{\ell-1}^{q(\ell-1)})| 
    \leq \sum_{\substack{j_1, j_2 \in [\ell-1]\\ q(j_1) = 2, q(j_2) = 3}}  \sum_{\sigma \in \mathcal{S}_{\ell-1}} \prod_{j \neq j_1, j_2} a_{\sigma(j), j} \cdot a_{\sigma(j_1), \ell}  \cdot a_{\ell, j_2}\cdot a_{\sigma(j_2), j_1} \cdot a_{\ell, \ell}^{-1} \notag\\
    &\leq \ell \cdot (\perm(A)-\prod_{i=1}^\ell a_{i,i}) \prod_{i=1}^\ell a_{i,i}\leq 0.05\cdot \prod_{i=1}^{\ell-1} a_{i,i}\,. \label{eq:24}
\end{align}
Plugging~\eqref{eq:18},~\eqref{eq:21},~\eqref{eq:23},~\eqref{eq:24} in~\eqref{eq:17}, we get
\begin{align}
    |\det(Z)| &\geq  |\det(z_1^0, z_2^0, \ldots, z_{\ell-1}^0)|-\sum_{f \in  \{\mathcal{Q}_2 \cup \mathcal{Q}_3 \cup \mathcal{Q}_4\} } 
   |\det(z_1^{q(1)}, z_2^{q(2)}, \ldots, z_{\ell-1}^{q(\ell-1)})| \notag \\
    &\geq 0.95\cdot \prod_{i=1}^{\ell-1} a_{i,i} - 0.3 \cdot \prod_{i=1}^{\ell-1} a_{i,i} = 0.65  \cdot \prod_{i=1}^{\ell-1} a_{i,i} \,. \label{eq:25}
\end{align}
From~\eqref{eq:12} and~\eqref{eq:25},
\begin{align*}
    \det(V_TV_T^\top) &\geq 0.65 \cdot a_{\ell, \ell}^2 \cdot (\prod_{i=1}^{\ell-1} a_{i,i})^2 \cdot \det(V_S V_S^\top)  >  0.65 \cdot f(\ell)^2 \cdot\det(V_S V_S^\top) > 2 \cdot \det(V_S V_S^\top)\,,
\end{align*}
where the second last inequality follows from part~\ref{item:a} of Claim~\ref{claim:f_viol}.
\end{proof}

\bibliographystyle{alpha}
\bibliography{ref}

\appendix
\section{Omitted Proofs and Examples} \label{appendix:intro}
\subsection{Matroids}
\begin{definition}[Matroid] A matroid $\mathcal{M}$ is a pair $(E, \mathcal{I})$, where $E$ is a finite set called the ground set of the matroid and $\mathcal{I}$ is a family of subsets of $E$ called the independent sets, satisfying the following properties:
\begin{itemize}
    \item The empty set is independent, i.e., $\emptyset \in \mathcal{I}$.
    \item Downward Closure: If $S \in \mathcal{I}$ and $T \subset S$, then $T \in \mathcal{I}$.
    \item Augmentation: If $S, T \in \mathcal{I}$ with $|T| < |S|$, then there exists $s \in S \backslash T$ such that $T\cup \{s\} \in \mathcal{I}$.
\end{itemize}
\end{definition}  
\begin{definition}[Basis of Matroid] An independent set $S \in \mathcal{I}$ is called a basis of a matroid $\mathcal{M} = (E, \mathcal{I})$ if it has the largest cardinality among all independent sets of $\mathcal{M}$.
\end{definition}
We crucially use the following fact about matroids.
\begin{fact}\label{fact:basis_exch}
For any two distinct bases $S$ and $T$ of matroid $\mathcal{M}$, there exists a bijection $h$ from $S$ to $T$, such that for every $s\in S\setminus T$, $(S\setminus \{s\})\cup \{h(s)\}$ is a basis of $\mathcal{M}$. 
\end{fact}

\subsection{Omitted Proofs from Section~\ref{sec:prelim}} \label{sec:omitted_proofs}
\begin{proof}[of Theorem~\ref{thm:main}]
Let $\OPT$ denote the optimal value of determinant maximization on vectors $v_1, \ldots, v_n$ with constraint matroid $\mathcal{M} = ([n], \mathcal{I})$.
By Theorem~\ref{thm:sparse-sol-gap}, there exists a fractional solution $\hat{x}$ to~\eqref{eq:cvx_program} with $|\supp(\hat{x})| \leq r + d^2 + 3d$ such that there exists a basis $T$ in $\M$ with $T\subset \supp(\hat{x})$ and
\begin{equation*}
   \det\left(\sum_{i\in T} v_iv_i^\top \right)  \geq (2e^5 d)^{-d} \cdot \OPT .
\end{equation*}

Let $E = \supp(\hat{x})$, and $\mathcal{M}_E = (E, \mathcal{I}_E)$ be the matroid $\mathcal{M}$ restricted to $E$. Let $\OPT_E$ be the optimal value of the determinant maximization problem on vectors in $\{v_i: i \in E\}$ with constraint matroid $\mathcal{M}_E$. Then $\OPT_E \det(\sum_{i \in T} v_i v_i^\top) \geq (2e^5 d)^{-d} \cdot \OPT$. 

Let $S$ be the current basis in Algorithm~\ref{alg:exch}.
Then by Lemma~\ref{lem:existence_cycle}, while
\begin{equation*}
    \det(\sum_{i \in S} v_i v_i^\top) < \OPT_E \cdot O(dk^{O(d)}) = \OPT_E \cdot O(d^{O(d)}),
\end{equation*}
there exists an $f$-violating cycle, $C$ in $G(S)$. So while  and we update $S \leftarrow S \triangle C$. So the basis $S$ returned by Algorithm~\ref{alg:exch} satisfies $\det(\sum_{i \in S} v_i v_i^\top) \geq \Omega(d^{-O(d)}) \cdot \OPT$\,.

If we initialize $S$ to any solution with a non-zero determinant, then in each iteration of Algorithm~\ref{alg:exch}, the ratio of the determinants satisfies
\begin{equation*}
   1 \leq \frac{\det(\sum_{i \in \OPT_E} v_i v_i^\top)}{\det(\sum_{i \in S} v_i v_i^\top)} \leq 2^{4\sigma}\,,
\end{equation*} 
 where $\sigma$ is the encoding length of the input to the determinant maximization problem (Chapter 3, Theorem 3.2 ~\cite{schrijver2000linear}). Theorem \ref{thm:main2} tells us that each iteration improves the objective by a factor of at least $2$, so Algorithm~\ref{alg:exch} will terminate in at most $4\sigma$ iterations.
\end{proof}

\section{Omitted Lemmas and proofs from Section~\ref{sec:existence}}\label{appendix:existence}
\begin{lemma}\label{lem:inv}
Let $S = \{v_1, \ldots, v_r\}$ such that $ V_S V_S^\top$ is invertible. Then \begin{itemize}
   \item $\det(I - V_Y^\top (V_S V_S^\top)^{-1} V_Y) \geq 0$ for any $Y \subseteq S$, 
    \item $|v_i^\top (V_S V_S^\top)^{-1} v_j| \leq \sqrt{(1-v_i^\top (V_S V_S^\top)^{-1} v_i) \cdot (1-v_j^\top (V_S V_S^\top)^{-1} v_j)}$ for any $i, j \in S$ and $i \neq j$.
\end{itemize}
\end{lemma}
\begin{proof}
For any $Y \subseteq S$, $V_SV_S^\top- V_YV_Y^\top = \sum_{v\in S\backslash Y} vv^\top \succeq 0$. Therefore,
\begin{align*}
    \det(V_S V_S^\top - V_YV_Y^\top) = \det(V_S V_S^\top) \cdot (1-V_Y^\top (V_S V_S^\top)^{-1} V_Y) \geq 0\,.
\end{align*}

For any $v_i, v_j \in S$ with $i \neq j$, $V_S V_S^\top- v_i v_i^\top -v_j v_j^\top \succeq 0$. Using the matrix determinant lemma gives
\begin{align*}
    \det(&V_S V_S^\top - v_i v_i^\top - v_j v_j^\top) 
    = \det\left(V_S V_S^\top + \begin{bmatrix}
    -v_i & -v_j
    \end{bmatrix} \begin{bmatrix}
    v_i^\top \\ v_j^\top
    \end{bmatrix}\right) \\
    &= \det(V_S V_S^\top) \cdot \left(I_2+ \begin{bmatrix}
    v_i^\top \\ v_j^\top
    \end{bmatrix}(V_S V_S^\top)^{-1}  \begin{bmatrix}
    -v_i & -v_j
    \end{bmatrix}\right) \\
    &=\det(V_S V_S^\top) \cdot \left((1-v_i^\top (V_S V_S^\top)^{-1} v_i) \cdot (1-v_j^\top (V_S V_S^\top)^{-1} v_j) -  (v_i^\top (V_S V_S^\top)^{-1} v_j)^2\right) \geq 0.
\end{align*}
This implies $|v_i^\top (V_S V_S^\top)^{-1} v_j| \leq \sqrt{(1-v_i^\top (V_S V_S^\top)^{-1} v_i) \cdot (1-v_j^\top (V_S V_S^\top)^{-1} v_j)}$\,.

\end{proof}
\section{Omitted Lemmas and proofs from Section~\ref{sec:update}} \label{appendix:update}
\begin{lemma}\label{lem:matroid_indep}
    If $C$ is a minimal $f$-violating cycle in $G(S)$, then $S\triangle C$ is independent in $\mathcal{M}$.
\end{lemma}
\begin{proof}
Let $T := S \triangle C$ and let $|C| = 2\ell$. Since $C$ is an $f$-violating cycle, $w(C) < - \log(f(\ell))$.

Define $N_1$ as the set of backward arcs in $C$ and $N_2$ as the set of forward arcs in $C$.
For the sake of contradiction, assume that $T$ is not independent in $\mathcal{M}$. Then, there exists a matching $N_1'$ on the vertices of $C$ consisting of only backward arcs such that $N_1 \neq N_1'$ (Chapter 39, Theorem 39.13, ~\cite{schrijver2003combinatorial}). Let $A$ be a multiset of arcs consisting of two copies of all arcs in $N_2$ plus all arcs in $N_1$ and $N_1'$ (with arcs in $N_1 \cap N_1'$ appearing twice). Consider the directed
graph $D = (\mathrm{ver}(C), A)$. Since $N_1 \neq N_1'$, $D$ contains
a directed circuit $C_1$ with $\mathrm{ver}(C_1) \subsetneq \mathrm{ver}(C)$. Every vertex in $\mathrm{ver}(C)$ has exactly two in-edges and two out-edges in $A$. Therefore, $D$ is Eulerian, and we can decompose $A$ into directed circuits $C_1, \ldots, C_k$. Since only arcs in $N_2$ have non-zero weights and the sum of these weights in $w(C)$, we have $\sum_{i=1}^k w(C_i) = 2w(C)$.

Because $\mathrm{ver}(C_1) \subsetneq \mathrm{ver}(C)$ by definition, at most one cycle in $C_1, \ldots, C_k$ can have the same set of vertices as $C$. If for some $j$, $\mathrm{ver}(C_j) = \mathrm{ver}(C)$, then $w(C_j) = w(C)$ as $C_j$ must contain every edge in $N_2$ once. So, $\sum_{i \neq j} w(C_i) = w(C) < - \log(f(\ell))$. Since $\sum_{i\neq j}|C_i|/2 = \ell$ and the function $f$ satisfies $f(a)\cdot f(b) < f(a+b)$, there must exist a cycle $C_i$ with $i\neq j$. Since only at most one cycle can have the same vertex set as $C$, we also have $\mathrm{ver}(C_i) \subsetneq \mathrm{ver}(C)$. So $C_i$ is an $f$-violating cycle with $\mathrm{ver}(C_i) \subsetneq \mathrm{ver}(C)$. This contradicts the minimality of $C$.

Otherwise, if $\mathrm{ver}(C_j) \subsetneq \mathrm{ver}(C)$ for all $j$, then $\sum_{i} w(C_i) = 2w(C) < -2\log(f(\ell))$.  Since $\sum_{i}|C_i|/2 = 2\ell$ and the function $f$ satisfies $f(a)\cdot f(b) < (f((a+b)/2))^2$ for any $a, b < (a+b)/2$, there exists a cycle $C_i$ such that $\mathrm{ver}(C_i) \subsetneq \mathrm{ver}(C)$ and $w(C_i) < -\log(f(|C_i|/2))$. So $C_i$ is an $f$-violating cycle with $\mathrm{ver}(C_i) \subsetneq \mathrm{ver}(C)$, which contradicts the fact that $C$ is a minimal $f$-violating cycle.
\end{proof}
\begin{lemma} \label{lem:local}
Let $S, T$ be sets of vectors in $\R^d$ such that $|S| = |T| = r$, $|S \triangle T| \leq 2d$, and let $X = T \backslash S$, $Y = S \backslash T$. Then \begin{equation*}
    \det(V_TV_T^\top) \geq \det(V_S V_S^\top) \cdot \det(V_X^\top (V_S V_S^\top)^{-1} V_Y)^2.
\end{equation*}
\end{lemma}
\begin{proof}
Let $\ell := |X| = |Y|$. Then by Lemma~\ref{lem:det_update}, we get
\begin{align}
    \det(V_TV_T^\top) &=  \det(V_S V_S^\top) \cdot \det\left(\begin{bmatrix} I_\ell + V_X^\top (V_S V_S^\top)^{-1}  V_X & V_X^\top (V_S V_S^\top)^{-1} V_X\\
   -V_Y^\top (V_S V_S^\top)^{-1} V_X &  I_\ell -V_Y^\top (V_S V_S^\top)^{-1}  V_X
    \end{bmatrix}\right)\,. \label{eq:26}
\end{align}
Since $I_\ell + V_X^\top (V_S V_S^\top)^{-1}  V_X \succ 0$, the $\ell \times \ell$ matrix $I_\ell + V_X^\top (V_S V_S^\top)^{-1}  V_X$ is invertible. 
We define
\begin{align*}
 A &:= I_\ell + V_X^\top (V_S V_S^\top)^{-1}  V_X, \quad \quad
 B :=V_X^\top (V_S V_S^\top)^{-1} V_Y, \quad\quad
 C := I_\ell -V_Y^\top (V_S V_S^\top)^{-1}  V_Y,\\
    Q &:= \begin{bmatrix} I_\ell & 0 \\
 B^\top A^{-1} & I_\ell\end{bmatrix}, \quad\quad
   P := \begin{bmatrix} I_\ell & -A^{-1} B \\
   0 & I_\ell\end{bmatrix}.
\end{align*}
Then \begin{align}
   Q \cdot & \begin{bmatrix} I_\ell + V_X^\top (V_S V_S^\top)^{-1}  V_X & V_X^\top (V_S V_S^\top)^{-1} V_Y\\
   -V_Y^\top (V_S V_S^\top)^{-1} V_X &  I_\ell -V_Y^\top (V_S V_S^\top)^{-1}  V_Y\end{bmatrix}\cdot P =\begin{bmatrix} A & 0\\
  0 &  C + B^\top A^{-1} B \end{bmatrix} . \label{eq:27}
\end{align}
Since $\det(P) = \det(Q) = 1$, \begin{align*}
  \det&\left( \begin{bmatrix} I_\ell + V_X^\top (V_S V_S^\top)^{-1}  V_X & V_X^\top (V_S V_S^\top)^{-1} V_Y\\
   -V_Y^\top (V_S V_S^\top)^{-1} V_X &  I_\ell -V_Y^\top (V_S V_S^\top)^{-1}  V_Y\end{bmatrix}\right) = \det\left( Q \cdot \begin{bmatrix} A & B\\
   -B^\top &  C\end{bmatrix}\cdot P\right)\\
   &= \det(A) \cdot \det( C + B^\top A^{-1} B ). \tag{from~\eqref{eq:27}}
\end{align*}
By Lemma~\ref{lem:inv}, $C = I_\ell -V_Y^\top (V_S V_S^\top)^{-1}  V_Y \succeq 0$, and since $B^\top A^{-1} B \succ 0$, we have $\det( C + B^\top A^{-1} B ) \geq \det( B^\top A^{-1} B )$. Therefore,
\begin{align}
  \det\left(  \begin{bmatrix} I_\ell + V_X^\top (V_S V_S^\top)^{-1}  V_X & V_X^\top (V_S V_S^\top)^{-1} V_Y\\
   -V_Y^\top (V_S V_S^\top)^{-1} V_X &  I_\ell -V_Y^\top (V_S V_S^\top)^{-1}  V_Y\end{bmatrix}\right)
   &\geq  \det(A) \det( B^\top A^{-1} B ) = \det(B)^2, \label{eq:28}
\end{align}
where the last equality holds because both $A$ and $B$ are $\ell \times \ell$ matrices.
Combining~\eqref{eq:26} and~\eqref{eq:28} completes the proof.
\end{proof}

\begin{proof}[of Claim~\ref{claim:f_viol}]\label{proof:claim_f_viol}
For part~\ref{item:a}, the product of the diagonal entries of $A$ is
\begin{equation*}
   \prod_{i=1}^\ell a_{i,i} \geq \left(\prod_{i=1}^{\ell-1} |\langle u_i, v_i \rangle_S| \right)\cdot  \sqrt{(1+\norm{u_\ell}_S^2) \cdot (1-\norm{v_\ell}_S^2)} \geq \exp(-w(C))\geq f(\ell). 
\end{equation*}
where the last inequality follows from the fact that $C$ is an $f$-violating cycle. 

For part~\ref{item:b}, if $a_{i,i} > |\dotp{i}{i}|$, then
\begin{equation*}
    \exp(-w(u_i \arc{2} v_j)) = \sqrt{(1+\normu{i}^2) \cdot (1-\normv{j}^2)} \geq |\dotp{i}{j}| = \exp(-w(u_i\arc{1} v_j)).
\end{equation*}

So replacing the arc $(u_i \arc{1} v_j)$ with $(u_i \arc{2} v_j)$ in $C$ gives an $f$-violating cycle $C'$, with $\mathrm{ver}(C) = \mathrm{ver}(C')$ but with two arcs of type $\RN{2}$. Since $C'$ contains two arcs of type $\RN{2}$, by Lemma~\ref{lem:typeII}, there exists an $f$-violating cycle $C''$ such that $\mathrm{ver}(C'') \subsetneq \mathrm{ver}(C') = \mathrm{ver}(C)$, which contradicts the minimality of $C$.

For part~\ref{item:c}, we first bound every off-diagonal entry of matrix $A$ as a function of its diagonal entries and then apply Lemma~\ref{lem:perm}.

For $i, j \in [\ell]$ with $i > j$, define the cycles $C^{\RN{1}}_{i,j} := (u_i \arc{1} v_j \rightarrow u_{j+1}\arc{1} v_{j+1} \ldots  v_{i-1} \rightarrow u_i)$ and $C^{\RN{2}}_{i,j} := (u_i \arc{2} v_j \rightarrow u_{j+1}\arc{1} v_{j+1}\ldots  v_{i-1} \rightarrow u_i)$. Both $C^{\RN{1}}_{i,j}$ and $C^{\RN{2}}_{i,j}$ contain $2(i-j)$ arcs and $\mathrm{ver}(C^{\RN{1}}_{i,j}) = \mathrm{ver}(C^{\RN{2}}_{i,j}) \subsetneq \mathrm{ver}(C)$. So, by minilality of $C$, $C^{\RN{1}}_{i,j}$ and $C^{\RN{2}}_{i,j}$ are not $f$-violating cycles.

Therefore, $\exp(-w(C^{\RN{1}}_{i,j})) = |\dotp{i}{j}| \cdot \prod_{k = j+1}^{i-1} \exp(-w(u_k \arc{1} v_k)) < f(i-j)$, and
\begin{equation}
    |\dotp{i}{j}| < \frac{f(i-j)}{\prod_{k = j+1}^{i-1}\exp(-w(u_k \arc{1} v_k))}\,. \label{eq:29}
\end{equation}
Using a similar argument for $C^{\RN{2}}_{i,j}$ gives
\begin{equation}
    \sqrt{(1+\normu{i}^2)\cdot (1-\normv{j}^2)} < \frac{f(i-j)}{\prod_{k = j+1}^{i-1} \exp(-w(u_k \arc{1} v_k))}\,.\label{eq:30}
\end{equation}
Combining~\eqref{eq:29} and~\eqref{eq:30}, we get
\begin{equation*}
    a_{i,j} \leq |\dotp{i}{j}| + \sqrt{(1+\normu{i}^2)\cdot (1-\normv{j}^2)} \leq \frac{2\cdot f(i-j)}{\prod_{k = j+1}^{i-1} \exp(-w(u_k \arc{1} v_k))} = \frac{2\cdot f(i-j)}{\prod_{k = j+1}^{i-1} |a_{k,k}|}\,,
\end{equation*}
where the last equality follows from part~\ref{item:b} of this claim.

For $i,j \in [\ell-1]$ with $i < j$, define $C^{\RN{1}}_{i,j} := (v_\ell \rightarrow u_1 \arc{1} v_1\ldots u_i \arc{1} v_j\rightarrow u_{j+1} \ldots  u_{\ell} \arc{2} v_{\ell})$ and $C^{\RN{2}}_{i,j} := (v_\ell \rightarrow u_1 \arc{1} v_1 \ldots u_i \arc{2} v_j\rightarrow u_{j+1} \ldots  u_{\ell} \arc{2} v_{\ell})$. Both $C^{\RN{1}}_{i,j}$ and $C^{\RN{2}}_{i,j}$ contain $2(\ell - j + i)$ arcs and $\mathrm{ver}(C^{\RN{1}}_{i,j}) = \mathrm{ver}(C^{\RN{2}}_{i,j}) \subsetneq \mathrm{ver}(C)$. So, they are not $f$-violating cycles. Therefore, \begin{align}
\exp(-w(C^{\RN{1}}_{i,j})) &= |\dotp{i}{j}| \cdot \prod_{k = 1}^{i-1} \exp(-w(u_k \arc{1} v_k)) \cdot \!\!\prod_{k = j+1}^{\ell-1} \exp(-w(u_k \arc{1} v_k)) \cdot \exp(-w(u_\ell \arc{2} v_\ell))\notag\\ 
&< f(\ell - j + i). \label{eq:31}
\end{align}
Since $C$ is an $f$-violating cycle, we also have \begin{equation}
    \prod_{k = 1}^{\ell-1} \exp(-w(u_k \arc{1} v_k)) \cdot \exp(-w(u_\ell \arc{2} v_\ell)) > f(\ell). \label{eq:32}
\end{equation}
Dividing~\eqref{eq:31} by~\eqref{eq:32} gives
 \begin{align}
     |\dotp{i}{j}| &< \frac{f(\ell-j+i)}{f(\ell)} \cdot \prod_{k=i}^j \exp(-w(u_k \arc{1} v_k))\,. \label{eq:33}
\end{align}
Similarly,
\begin{align}
     \sqrt{(1+\normu{i}^2)\cdot (1-\normv{j}^2)} &< \frac{f(\ell-j+i)}{f(\ell)} \cdot \prod_{k=i}^j \exp(-w(u_k \arc{1} v_k))\,. \label{eq:34}
\end{align}
Summing~\eqref{eq:33} and~\eqref{eq:34} gives
\begin{align*}
    a_{i,j} &\leq |\dotp{i}{j}| + \sqrt{(1+\normu{i}^2)\cdot (1-\normv{j}^2)} \leq 2\cdot\frac{f(\ell-j+i)}{f(\ell)} \cdot \prod_{k=i}^j \exp(-w(u_k \arc{1} v_k))\\
    &\leq 2\cdot \frac{ f(\ell-j+i)}{f(\ell)} \cdot \prod_{k=i}^j |a_{k,k}|\,. \tag{from definition of $a_{k, k}$}
\end{align*}
To bound $a_{i, \ell}$ for $i < \ell$, consider cycles $C^{\RN{1}}_{i,\ell} := (v_\ell \rightarrow u_1 \arc{1} v_1\ldots u_i \arc{1} v_\ell)$ and $C^{\RN{2}}_{i,j} := (v_\ell \rightarrow u_1 \arc{1} v_1 \ldots u_i \arc{2} v_\ell)$. Both $C^{\RN{1}}_{i,\ell}$ and $C^{\RN{2}}_{i,\ell}$ contain $2i $ arcs and $\mathrm{ver}(C^{\RN{1}}_{i,\ell}) = \mathrm{ver}(C^{\RN{2}}_{i,\ell}) \subsetneq \mathrm{ver}(C)$. So, they are not $f$-violating cycles. Following a similar argument to the $i < j < \ell$ case and comparing $w(C^{\RN{1}}_{i,\ell}), w(C^{\RN{1}}_{i,\ell})$ with $w(C)$ gives \begin{equation*}
    a_{i, \ell} \leq 2\cdot \frac{f(i)}{f(\ell)} \cdot \prod_{k=i}^{\ell-1} |a_{k,k}| \cdot \exp(-w(u_\ell \arc{2} v_\ell)) < 2\cdot \frac{f(i)}{f(\ell)} \cdot \prod_{k=i}^{\ell-1} |a_{k,k}| \cdot a_{\ell, \ell}\,.
\end{equation*}

So $A$ satisfies all three prerequisites of Lemma~\ref{lem:perm}. Applying Lemma~\ref{lem:perm} to $A$ gives $\perm(A) \leq \prod_{i=1}^\ell a_{i,i} (1+0.05/\ell) $.

For any $t \in [\ell]$, the principal submatrix $A_{t,t}$ satisfies  prerequisites~\ref{th:first} and~\ref{th:second} of Lemma~\ref{lem:perm}. Observe that $f(y - x) /f(y)$ is a non-increasing function of $y$ for any fixed $x \in [\ell]$. Therefore, for any $i > j$,
\begin{equation*}
    a_{i,j} \leq\frac{2f(\ell-j+i)}{f(\ell)} \cdot \prod_{k=i}^j a_{i,i}  \leq \frac{2f(t-j+i)}{f(t)} \cdot \prod_{k=i}^j a_{i,i}\,,
\end{equation*}
as long as $t \leq \ell$. So $A_{t,t}$ also satisfies prerequisite~\ref{th:third} of Lemma~\ref{lem:perm}. Therefore, by Lemma~\ref{lem:perm}, $\perm(A_{t,t}) \leq \prod_{i=1}^t a_{i,i} \cdot( 1+0.05/t)$.

\end{proof}
\begin{proof}[of Claim~\ref{claim:z}] \label{proof:claim_z}For any $i, j \in [\ell-1]$, the $i$-th entry of $|z_j^1|$ can be bounded as
\begin{align*}
    |z_j^1(i)| &= a_{\ell,\ell}^{-2}\cdot(1-\normv{\ell}^2) \cdot |(u_i^\top (V_S V_S^\top)^{-1} u_\ell) \cdot \dotp{\ell}{j}|\\
    &\leq (1-\normv{\ell}^2) \cdot |u_i^\top (V_S V_S^\top)^{-1} u_\ell| \cdot a_{\ell, j}\cdot a_{\ell, \ell}^{-2}\,,
\end{align*}
where the last inequality follows from the definition of $a_{i, j}$.
By the Cauchy Schwarz inequality, $|u_i^\top (V_S V_S^\top)^{-1} u_\ell| \leq \sqrt{(u_i^\top (V_S V_S^\top)^{-1} u_i) \cdot (u_\ell^\top (V_S V_S^\top)^{-1} u_\ell)} = \normu{i}\cdot \normu{\ell}$. Therefore, 
\begin{align*}
|z_j^1(i)| &\leq (1-\normv{\ell}^2) \cdot\normu{i}\cdot \normu{\ell} \cdot a_{\ell, j} \cdot a_{\ell, \ell}^{-2} \\
&= \sqrt{\normu{i}^2 \cdot(1-\normv{\ell}^2) } \cdot  \sqrt{\normu{\ell}^2 \cdot(1-\normv{\ell}^2) }\cdot a_{\ell, j} \cdot a_{\ell, \ell}^{-2} \\
&\leq a_{i, \ell} \cdot a_{\ell, \ell}\cdot a_{\ell, j} \cdot a_{\ell, \ell}^{-2} =  a_{i, \ell}\cdot a_{\ell, j} \cdot a_{\ell, \ell}^{-1}\,.  \tag{since $\sqrt{\normu{i}^2\cdot (1-\normv{j}^2)} \leq a_{i,j}$}
\end{align*}

Similarly, the $i$-th entry of $|z_j^2|$ is given by
\begin{align*}
    |z_j^2(i)| &= a_{\ell,\ell}^{-2}\cdot\dotp{\ell}{\ell} \cdot |(u_i^\top (V_S V_S^\top)^{-1} u_\ell) \cdot \dotv{\ell}{j}| \\ &\leq  |u_i^\top (V_S V_S^\top)^{-1} u_\ell| \cdot |v_\ell^\top  (V_S V_S^\top)^{-1} v_j|  \cdot a_{\ell, \ell}^{-1}  \tag*{(since $|\dotp{\ell}{\ell}|  \leq a_{\ell, \ell}$)}
    \\&\leq  \normu{i} \cdot \normu{\ell} \cdot |v_\ell^\top  (V_S V_S^\top)^{-1} v_j|  \cdot a_{\ell, \ell}^{-1} . \tag*{(by Cauchy Schwarz)}
\end{align*}
By Lemma~\ref{lem:inv}, $|v_\ell^\top (V_S V_S^\top)^{-1} v_j| \leq \sqrt{(1-\normv{\ell}^2)\cdot (1-\normv{j}^2)}$. Therefore,
\begin{align*}
    |z_j^2(i)| &\leq \normu{i}\cdot \normu{\ell} \cdot \sqrt{(1-\normv{\ell}^2) \cdot (1-\normv{j}^2)} \cdot a_{\ell, \ell}^{-1} \\
    &=\sqrt{ \normu{i}^2\cdot (1-\normv{\ell}^2) } \cdot \sqrt{\normu{\ell}^2 \cdot (1-\normv{j}^2)} \cdot a_{\ell, \ell}^{-1}
    \leq a_{i,\ell} \cdot a_{\ell, j} \cdot a_{\ell, \ell}^{-1}.
\end{align*}
Similarly,
\begin{align*}
    |z_j^3(i)| &= a_{\ell,\ell}^{-2}\cdot\dotp{\ell}{\ell} \cdot |( u_i^\top (V_S V_S^\top)^{-1} v_\ell) \cdot \dotp{\ell}{j}| \leq a_{i,\ell} \cdot a_{\ell, j} \cdot a_{\ell, \ell}^{-1} \quad \mathrm{and}\\
    |z_j^4(i)| &= a_{\ell,\ell}^{-2}\cdot (1+\normu{\ell}^2) \cdot |(u_i^\top (V_S V_S^\top)^{-1} v_\ell)| \cdot \dotv{\ell}{j} \leq a_{i, \ell} \cdot a_{\ell, j} \cdot a_{\ell, \ell}^{-1}\,.
\end{align*}
This concludes the proof of the first part of the claim.

For columns $j_1, j_2$ with $j_1 \neq j_2$, and rows $i_1, i_2$ with $i_1 \neq i_2$, we have
\begin{align*}
    |z_{j_1}^1(i_1) \cdot z_{j_2}^4(i_2)|  &= a_{\ell,\ell}^{-4}\cdot  (1-\normv{\ell}^2) \cdot |\dotu{i_1}{\ell}\cdot \dotp{\ell}{j_1}|
    \cdot (1+\normu{\ell}^2) \cdot |\dotp{i_2}{\ell} \cdot \dotv{\ell}{j_2}| \\
    &\leq a_{\ell,\ell}^{-2}\cdot |\dotu{i_1}{\ell}|\cdot |\dotp{\ell}{j_1}|
    \cdot |\dotp{i_2}{\ell}| \cdot |\dotv{\ell}{j_2}| \tag{since $(1+\normu{\ell}^2)\cdot(1-\normv{\ell}^2)\leq a_{\ell,\ell}^2 $}\\
    &\leq a_{\ell,\ell}^{-2}\cdot |\dotu{i_1}{\ell} | \cdot a_{\ell,j_1} \cdot a_{i_2, \ell} \cdot  |\dotv{\ell}{j_2}| \tag{since $|\dotp{i}{j}| \leq a_{i,j}$}\\
    &\leq a_{\ell, \ell}^{-2}\cdot \normu{i_1} \cdot \normu{\ell} \cdot a_{\ell,j_1} \cdot a_{i_2, \ell} \cdot |\dotv{\ell}{j_2}| \,. \tag{by Cauchy Schwarz}
\end{align*}
By Lemma~\ref{lem:inv}, $|\dotv{\ell}{j_2}| \leq \sqrt{(1-\normv{\ell}^2) \cdot (1-\normv{j_2}^2)}$. Therefore,
\begin{align*}
     |z_{j_1}^1(i_1) \cdot z_{j_2}^4(i_2)|  &\leq \normu{i_1} \cdot \normu{\ell} \cdot a_{\ell,j_1} \cdot a_{i_2, \ell} \cdot \sqrt{(1-\normv{\ell}^2) \cdot (1-\normv{j_2}^2)} \cdot a_{\ell, \ell}^{-2}\\
     &= a_{\ell,j_1} \cdot a_{i_2, \ell} \cdot \sqrt{ \normu{\ell}^2  \cdot (1-\normv{\ell}^2)} \cdot \sqrt{\normu{i_1}^2 \cdot(1-\normv{j_2}^2)} \cdot a_{\ell, \ell}^{-2}\\
    &\leq a_{i_1, j_2}  \cdot a_{\ell, j_1}\cdot a_{i_2, \ell} \cdot a_{\ell, \ell}^{-1}\,. \tag{since $ \sqrt{ \normu{\ell}^2  \cdot (1-\normv{\ell}^2)}  \leq a_{\ell, \ell}$}
\end{align*}
By a similar argument, one can conclude $|z_{j_1}^2(i_1) \cdot z_{j_2}^3(i_2)|  \leq a_{i_1, \ell}  \cdot a_{\ell, j_2}\cdot a_{i_2, j_1} \cdot a_{\ell, \ell}^{-1}$\,.

\end{proof}

\section{Permanent of the Coefficient Matrix} \label{appendix:perm}
\permanent*
\begin{proof}
We first scale $A_\ell$ so that the upper bounds on off-diagonal entries are independent of the diagonal entries. Let $B_\ell$ be the matrix obtained by applying the following operation to $A_\ell$:
\begin{itemize}
    \item Divide the $j$-th column by $\prod_{x=1}^j a_{x,x}$ for all $j \in [\ell]$ 
    \item Multiply the $i$-row by $\prod_{x=1}^{i-1} a_{x,x}$ and for $i > 1$
\end{itemize}
Then $\perm(A_\ell) = \prod_{i=1}^\ell a_{i,i} \cdot \perm(B_\ell) $, and the entries of $B_\ell$ satisfy
\begin{enumerate}
    \item \label{item:one} $b_{i,i} = 1$ for all $i \in [\ell]$,
    \item \label{item:two} $b_{i,j} \leq 2 \cdot f(i-j)$ for all $j < i$, and
    \item  \label{item:three}$b_{i,j} \leq  2 \cdot \frac{f(\ell - j + i)}{f(\ell)} $ for all $i < j $.
\end{enumerate}
For $\ell = 2$, $\perm(B_2) = 1 + b_{1,2} \cdot b_{2, 1} \leq 1 + \frac{4f(1)^2}{f(2)} < 1.006$ and therefore $\perm(A_2) \leq (1+0.05/2) \cdot a_{1,1}\cdot a_{2,2}$.

So we assume $\ell \geq 3$. Let $\mathcal{S}_\ell$ denote the set of permutations on $[\ell]$, and $\id_\ell$ denote the identity permutation on $\ell$ elements. Expanding the permanent of $B_\ell$ gives \begin{equation}
    \perm(B_\ell) =\sum_{\sigma \in \mathcal{S}_\ell} \prod_{i=1}^\ell b_{i, \sigma(i)} = 1 + \sum_{\sigma \in \mathcal{S}_\ell \backslash \{\id_\ell\}} \prod_{i=1}^\ell b_{i, \sigma(i)} \,. \label{eq:35}
\end{equation}
We categorize all permutations in $\mathcal{S}_\ell \backslash \{\id_\ell\}$ based on the number of \emph{fixed points} and \emph{exceedances}. The set of fixed points of a permutation $\sigma \in \mathcal{S}_\ell$ is defined as $\{i \in [\ell]: \sigma(i) = i\}$ and the exceedance of $\sigma$ is defined as the number of indices $i$ such that $\sigma(i) > i$ (for more details, see Definition~\ref{lem:exc} and Lemma~\ref{lem:derangements}). Let $\mathcal{S}_\ell(n, k)$ denote the subset of $\mathcal{S}_\ell$ with $\ell - n$ fixed points and $k$ exceedances (see Figure~\ref{fig:permutation_types}).
\begin{figure}[ht]
    \centering
    \includegraphics[width=0.6\textwidth]{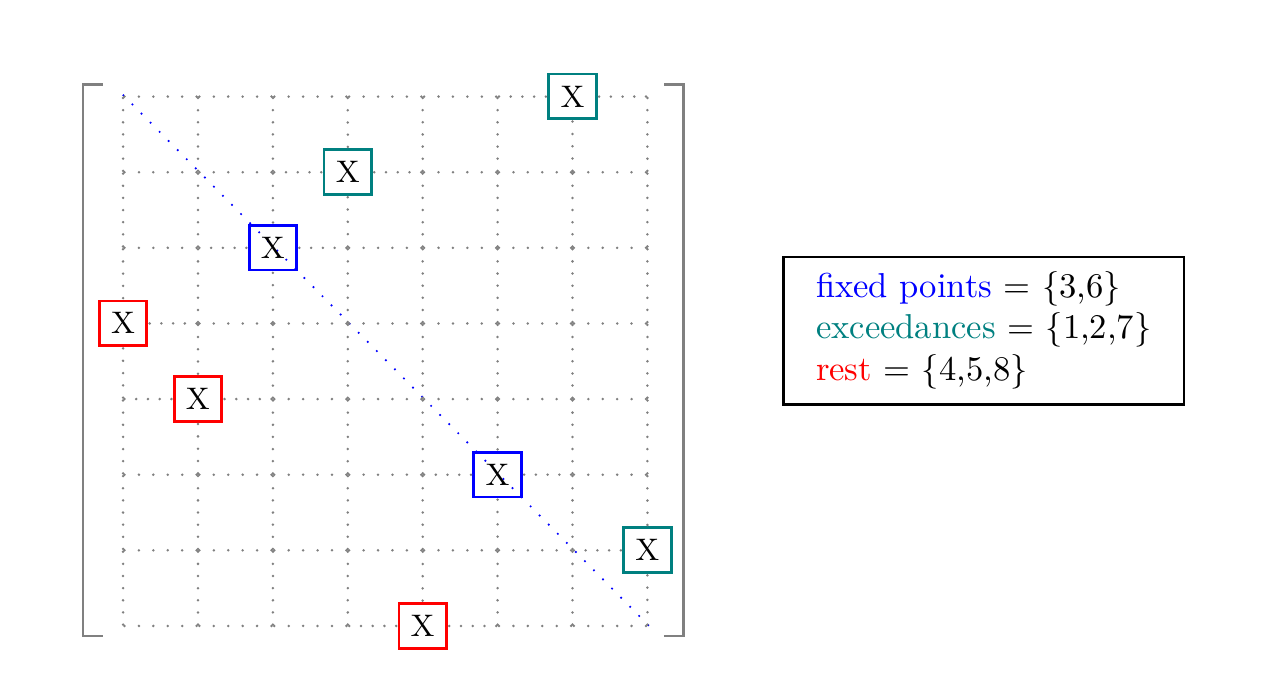}
    \caption{A permutation and corresponding index partition in $\mathcal{S}_8(6,3)$.}
    \label{fig:permutation_types}
\end{figure}

By definition, all permutations in $\mathcal{S}_\ell \backslash \{\id_\ell\}$ have at most $\ell-2$ fixed points and at least $1$ exceedance. So we can further expand~\eqref{eq:35} as
\begin{equation}
    \perm(B_\ell) = 1 +  \sum_{n=2}^\ell \; \sum_{k=1}^{n-1} \sum_{\sigma \in \mathcal{S}_\ell(n, k)} \prod_{i=1}^\ell b_{i, \sigma(i)}  \,. \label{eq:36}
\end{equation}

For a permutation $\sigma \in \mathcal{S}_\ell(n, k)$,  
\begin{align*}
    \prod_{i=1}^\ell b_{i, \sigma(i)} &= \prod_{i = \sigma(i)} b_{i, i} \cdot \prod_{i > \sigma(i)} b_{i, \sigma(i)} \cdot \prod_{i < \sigma(i)} b_{i, \sigma(i)} = \prod_{i > \sigma(i)} b_{i, \sigma(i)} \cdot \prod_{i < \sigma(i)} b_{i, \sigma(i)} \tag{since $b_{i,i} = 1$}\\
    &\leq  \prod_{i > \sigma(i)} 2 f(i-\sigma(i)) \cdot \prod_{i < \sigma(i)} \frac{2f(\ell-\sigma(i)+i)}{f(\ell)}\,. \tag{from properties~\ref{item:two} and~\ref{item:three} }
\end{align*}

Since $\sum_{i=1}^\ell i-\sigma(i) = 0$ for any permutation,
\begin{equation*}
    \sum_{i > \sigma(i)}  i - \sigma(i) + \sum_{i < \sigma(i)} \ell - \sigma(i) + i  = \ell \cdot |\{i : \sigma(i) > i\}| = \ell \cdot k
r\end{equation*} for any $\sigma \in \mathcal{S}_\ell(n, k)$. Therefore, if $\sigma \in \mathcal{S}_\ell(n, k)$, then there exist integers $1 \leq x_1, x_2, \ldots, x_n \leq \ell-1$ with $\sum_{i=1}^n x_i = \ell \cdot k$, such that
\begin{align}
    \prod_{i=1}^\ell b_{i, \sigma(i)} \leq 2^n \cdot \frac{\prod_{i=1}^n f(x_i)}{f(\ell)^k}\,. \label{eq:37}
\end{align}

If $k \geq n/2$, and $x_1, \ldots, x_n$ satisfy $\sum_{i=1}^n x_i = \ell \cdot k$, and $1\leq x_i \leq \ell-1$, then $\prod_{i=1}^n f(x_i)$ is maximized when $x_1 = \ldots = x_k = \ell-1$, $x_{k+1} = 2k-n+1$, and $x_{k+2} = \ldots = x_n = 1$ because $f$ satisfies $f(a+b) \geq f(a) \cdot f(b)$. Similarly, when $k < n/2$,  $\prod_{i=1}^n f(x_i)$ is maximized when $x_1 = \ldots = x_{k-1} = \ell-1$, $x_{k} = \ell-n+2k-1$, and $x_{k+1} = \ldots = x_n = 1$. So,
\begin{equation*}
    \prod_{i=1}^n f(x_i) \leq \begin{cases} 
    f(\ell-1)^k \cdot f(2k-n+1) \cdot f(1)^{n-k-1} & \text{ if } k \geq n/2\\
    f(\ell-1)^{k-1} \cdot f(\ell-n+2k-1) \cdot f(1)^{n-k} & \text{ otherwise.}
    \end{cases}
\end{equation*}
Plugging these bounds in~\eqref{eq:37}, we get
\begin{align*}
    \prod_{i=1}^\ell b_{i, \sigma(i)} \leq \begin{cases} 
    2^n \cdot \frac{f(\ell-1)^k \cdot f(2k-n+1) \cdot f(1)^{n-k-1}}{f(\ell)^k} & \text{ if } k \geq n/2\\
    2^n \cdot \frac{f(\ell-1)^{k-1} \cdot f(\ell-n+2k-1) \cdot f(1)^{n-k}}{f(\ell)^k} & \text{ if } k < n/2\,.
    \end{cases} 
\end{align*}
Using the definition of $f$, for any $k \geq n/2$,
\begin{align}
    \prod_{i=1}^\ell b_{i, \sigma(i)} \leq 2^{2n-k-1} \cdot \frac{((2k-n+1)!)^{11}}{\ell^{11k}} \, , \label{eq:38}
\end{align}
and for $k < n/2$,
\begin{align}
    \prod_{i=1}^\ell b_{i, \sigma(i)} \leq 2^{2n-k} \cdot \frac{1}{\ell^{11(k-1)} \cdot (\ell \cdot (\ell-1) \ldots (\ell-n+2k))^{11} } \, .\label{eq:39}
\end{align}

We will now bound the sum of $\prod_{i=1}^\ell b_{i, \sigma(i)}$ over all permutations $\sigma$ in $\mathcal{S}_\ell(n,k)$ based on the relative values of $n$ and $k$. 

For fixed $k$ and $n$ with $k < n/2$, summing over all permutations in $\mathcal{S}_\ell(n,k)$ gives
\begin{align}
    \sum_{\sigma \in \mathcal{S}_\ell(n, k)} \prod_{i=1}^\ell b_{i, \sigma(i)} &\leq |\mathcal{S}_\ell(n,k)| \cdot \max_{\sigma \in \mathcal{S}_\ell(n,k)} \prod_{i=1}^\ell b_{i, \sigma(i)} =  \binom{\ell}{n} \cdot \theta(n,k) \cdot \max_{\sigma \in \mathcal{S}_\ell(n,k)} \prod_{i=1}^\ell b_{i, \sigma(i)} \tag{from Lemma~\ref{lem:derangements}}\\
    &\leq \binom{\ell}{n} \cdot  \theta(n,k) \cdot 2^{2n-k} \cdot \frac{1}{\ell^{11(k-1)} \cdot (\ell \cdot (\ell-1) \ldots (\ell-n+2k))^{11} } \, , \label{eq:40}
\end{align}
where the last inequality follows from~\eqref{eq:39}.
We will bound the three terms of equation~\eqref{eq:40}, $\binom{\ell}{n}$, $ \theta(n,k)$, and $2^{2n-k}$ separately.

Expanding the first term, we get 
\begin{align}
   \binom{\ell}{n} \cdot \frac{1}{\ell^{(k-1)} \cdot (\ell \cdot (\ell-1) \ldots (\ell-n+2k))} = \frac{\ell \cdot (\ell-1) \ldots (\ell-n+1)}{n! \cdot \ell^{(k-1)} \cdot (\ell \cdot (\ell-1) \ldots (\ell-n+2k))} < \frac{1}{n! \cdot \ell^{k-1}} \,. \label{eq:41}
\end{align}
For the third term, note that $k< n/2$ implies that $2^{2n-k} < 2^{3n-3k}$, and since $k \geq 1$, $\ell -n +2k > 2$. Using these two facts, we get 
\begin{equation}
     2^{2n-k} \cdot \frac{1}{\ell^{3(k-1)} \cdot (\ell \cdot (\ell-1) \ldots (\ell-n+2k))^{3} } < \frac{2^{3(n-k)}}{(\ell - n + 2k)^{3(n-k)}} < 1 \,. \label{eq:42}
\end{equation} 

Plugging~\eqref{eq:41} and~\eqref{eq:42} in~\eqref{eq:40}, we get
\begin{align}
    \sum_{\sigma \in \mathcal{S}_\ell(n, k)} \prod_{i=1}^\ell b_{i, \sigma(i)} \leq \frac{1}{n!\cdot \ell^{k-1}} \cdot \theta(n,k)  \cdot \frac{1}{\ell^{7(k-1)} \cdot (\ell \cdot (\ell-1) \ldots (\ell-n+2k))^{7} } \, .\label{eq:43}
\end{align}

For $k = 1$, $\theta(n, k) = 1$ by Lemma~\ref{lem:derangements}, and therefore \begin{align}
    \sum_{\sigma \in \mathcal{S}_\ell(n, 1)} \prod_{i=1}^\ell b_{i, \sigma(i)} &\leq  \frac{1}{n!} \cdot  \frac{1}{(\ell \cdot (\ell-1) \ldots (\ell-n+2))^{7}}\,. \label{eq:44}
\end{align}
Summing over all $n$ gives
 \begin{align}
    \sum_{n=2}^{\ell}\sum_{\sigma \in \mathcal{S}_\ell(n, 1)} \prod_{i=1}^\ell b_{i, \sigma(i)} &\leq  \sum_{n=2}^{\ell}\frac{1}{n!} \cdot  \frac{1}{(\ell \cdot (\ell-1) \ldots (\ell-n+2))^{7}} \leq \frac{1}{2\ell^6}\,. \label{eq:45}
\end{align}
For any $2 \leq k < n/2$,  using Lemma \ref{lem:derangements}, we have $\theta(n, k) = (2k+3)^{n}$. Since $k \geq 2$, $2k+3 \leq 2 \cdot 2k$, and therefore $\theta(n, k) \leq (2 \cdot 2k)^{n+2}$. Plugging this is~\eqref{eq:43}, we get 
\begin{align}
    \sum_{\sigma \in \mathcal{S}_\ell(n, k)} \prod_{i=1}^\ell b_{i, \sigma(i)} &\leq \frac{1}{n! \cdot \ell^{k-1}} \cdot (2k)^{n} \cdot 2^{n} \cdot \frac{1}{\ell^{7(k-1)} \cdot (\ell \cdot (\ell-1) \ldots (\ell-n+2k))^{7} }\,. \label{eq:46}
\end{align}

Moreover $k < n/2$ implies that $n-k > n/2$, and as a result  $n \leq 2(n-k)$ and $ 2^{n} \cdot (2k)^{n} \leq 2^{2(n-k)} \cdot (2k)^{2(n-k)}$. Therefore,
\begin{align}
   2^{n} \cdot (2k)^{n}  &\cdot \frac{1}{\ell^{4(k-1)} \cdot (\ell \cdot (\ell-1) \ldots (\ell-n+2k))^4} \notag \\ &<  2^{2(n-k)} \cdot (2k)^{2(n-k)}  \cdot \frac{1}{\ell^{4(k-1)} \cdot (\ell \cdot (\ell-1) \ldots (\ell-n+2k))^4} \notag \\
    &< \left(\frac{2}{\ell-n+2k}\right)^{2(n-k)}   \cdot \left(\frac{2k}{\ell-n+2k}\right)^{2(n-k)} < 1, \label{eq:47}
\end{align}
where the last inequality follows from  $2k \leq \ell -n + 2k$ and $2 \leq \ell -n + 2k$.

Combining~\eqref{eq:46}and~\eqref{eq:47}, we have for any $2 \leq k < n/2$,
\begin{equation}
     \sum_{\sigma \in \mathcal{S}_\ell(n, k)} \prod_{i=1}^\ell b_{i, \sigma(i)} \leq  \frac{1}{n!\cdot \ell^{4(k-1)} \cdot (\ell \cdot (\ell-1) \ldots (\ell-n+2k))^3} \,. \label{eq:48}
\end{equation}
Summing over all $k$ and $n$ with $2 \leq k < n/2$ gives
\begin{align}
    \sum_{n=2}^{\ell} \sum_{k=2}^{\lceil n/2\rceil -1}\sum_{\sigma \in \mathcal{S}_\ell(n, k)} \prod_{i=1}^\ell b_{i, \sigma(i)} &=  \sum_{n=5}^{\ell} \sum_{k=2}^{\lceil n/2\rceil -1}\sum_{\sigma \in \mathcal{S}_\ell(n, k)} \prod_{i=1}^\ell b_{i, \sigma(i)} \tag{since $n > 2k > 4$} \\
    &\leq  \sum_{n=5}^{\ell} \sum_{k=2}^{\lceil n/2\rceil -1}\frac{1}{n!\cdot \ell^{4(k-1)} \cdot (\ell \cdot (\ell-1) \ldots (\ell-n+2k))^3} \tag{from~\eqref{eq:48}}\\
    &\leq \frac{1}{5!} \sum_{n=5}^{\ell} \sum_{k=2}^{\lceil n/2\rceil -1}\frac{1}{\ell^{4(k-1)}} \leq  \frac{1}{5!} \sum_{n=5}^{\ell} \sum_{k=2}^{\lceil n/2\rceil -1}\frac{1}{\ell^{4}} \tag{since $k > 2$}\\
    &\leq \frac{1}{5!} \cdot \ell \cdot \lfloor \ell/2 \rfloor \cdot \frac{1}{\ell^{4}} \leq \frac{1}{2\cdot 5! \cdot \ell^2}\,. \label{eq:49}
\end{align}

Similarly, for a fixed $k$ and $n$ with $k \geq n/2$, summing over all permutations in $\mathcal{S}_\ell(n,k)$,
\begin{align}
    \sum_{\sigma \in \mathcal{S}_\ell(n, k)} \prod_{i=1}^\ell b_{i, \sigma(i)} &\leq \binom{\ell}{n} \cdot \theta(n,k) \cdot 2^{2n-k-1} \cdot \frac{((2k-n+1)!)^{11}}{\ell^{11k}}\,. \label{eq:50}
\end{align}
We again bound the three terms, namely $\binom{\ell}{n}$, $\theta(n,k)$, and $2^{2n-k-1}$ separately.

For the first term, since $n/2 \leq k$, $2n-k-1 \leq 3k-1$, and therefore
\begin{align*}
  2^{2n-k-1}  \cdot \frac{((2k-n+1)!)^3}{\ell^{3k}} &\leq  2^{3k-1}  \cdot  \frac{((2k-n+1)!)^3}{\ell^{3k}}\,.
\end{align*}
Since $2k-n+1 \leq k$, and $k+1 \leq \ell$, $2^{3k-1}  \cdot \frac{((2k-n+1)!)^3}{\ell^{3k}} \leq  \frac{1}{2}  \cdot  \frac{(2^k \, k!)^3}{(k+1)^{3k}}$.

For $k = 1, 2, 3, 4, 5$, $\frac{(2^k k!)^3}{(k+1)^{3k}} \leq 1$. For $k \geq 6$, $2^k k! \leq k^k$. Therefore, 
\begin{align}
  2^{2n-k-1}  \cdot \frac{((2k-n+1)!)^3}{\ell^{3k}} &\leq  \frac{1}{2} \,. \label{eq:51}
\end{align}

Expanding the third term, 
\begin{align*}
   \binom{\ell}{n} \cdot \frac{((2k-n+1)!)^2}{\ell^{2k}} &= \frac{\ell\cdot (\ell-1) \ldots (\ell-n+1)}{n!} \cdot \frac{((2k-n+1)!)^2}{\ell^{2k}}\\
   &\leq \frac{1}{n!} \cdot \frac{((2k-n+1)!)^2}{\ell^{2k-n}} = \frac{1}{\ell (n-1)!} \cdot \frac{((2k-n+1)!)^2}{n\ell^{2k-n-1}} \leq \frac{(2k-n+1)!}{\ell (n-1)!} \,.
\end{align*}
Since $k+1 \leq n$, we have $2k-n+1 \leq n-1$, and therefore
\begin{align}
   \binom{\ell}{n} \cdot \frac{((2k-n+1)!)^2}{\ell^{2k}} &\leq \frac{(2k-n+1)!}{l(n-1)!} \leq \frac{1}{\ell} \,. \label{eq:52}
\end{align}

Plugging in~\eqref{eq:51} and~\eqref{eq:52} in~\eqref{eq:50}, 
\begin{align}
    \sum_{\sigma \in \mathcal{S}_\ell(n, k)} \prod_{i=1}^\ell b_{i, \sigma(i)} &\leq  \frac{1}{2\ell} \cdot \theta(n,k) \cdot \frac{((2k-n+1)!)^6}{\ell^{6k}} \,. \label{eq:53}
\end{align}
Since $\theta(n, n-1) = 1$, for $k = n-1$, we have 
\begin{align}
    \sum_{\sigma \in \mathcal{S}_\ell(n, n-1)} \prod_{i=1}^\ell b_{i, \sigma(i)} &\leq  \frac{1}{2\ell} \cdot \frac{((n-1)!)^6}{\ell^{6(n-1)}} \,. \label{eq:54}
\end{align}
Summing over all $n$ gives
\begin{align}
    \sum_{n=2}^{\ell}\sum_{\sigma \in \mathcal{S}_\ell(n, n-1)} \prod_{i=1}^\ell b_{i, \sigma(i)} &\leq  \sum_{n=2}^{\ell}\frac{1}{2\ell} \cdot \frac{((n-1)!)^6}{\ell^{6(n-1)}} \leq \sum_{n=2}^{\ell}\frac{1}{2\ell} \cdot \frac{1}{\ell^{6}}  \tag{since $n-1 < \ell$} \\
    &\leq \frac{1}{2\ell^6}  \label{eq:55}
\end{align}
By Lemma \ref{lem:derangements}, for $k < n-1$, $\theta(n, k) \leq  (2n-2k+5)^{n}$ for $k \geq n/2$, and
\begin{align*}
    \sum_{\sigma \in \mathcal{S}_\ell(n, k)} \prod_{i=1}^\ell b_{i, \sigma(i)} &\leq  \frac{1}{2\ell} \cdot (2n-2k+5)^{n} \cdot \frac{((2k-n+1)!)^6}{\ell^{6k}}\,.
\end{align*}
Let $z = 2k - n$, then
\begin{align*}
  (2n-2k+5)^{n} &\cdot \frac{((2k-n+1)!)^6}{\ell^{6k}} = (2k-2z+5)^{2k-z}  \cdot \frac{((z+1)!)^6}{\ell^{6k}} \\
  &\leq (2k-z+5)^{2k-z}  \cdot \frac{((z+1)!)^6}{\ell^{6k}} \,.
\end{align*}
Taking the partial derivative of $ \frac{(2k-z+5)^{2k-z}}{\ell^{6k}}$ with respect to $k$,
\begin{align*}
   \frac{\partial \frac{(2k-z+5)^{2k-z}}{\ell^{6k}}}{\partial k} &= \frac{(2k-z+5)^{2k-z}}{\ell^{6k}}\left(2\cdot \log(2k-z+5) + 2\cdot\frac{2k-z}{2k-z+5} -6\log(\ell)\right)\,,
\end{align*}
where the last inequality follows since $2k-z \leq \ell$ and $\ell \geq 3$.

Therefore $ \frac{(2k-z+5)^{2k-z+2}}{\ell^{6k}}$ is a non-increasing function of $k$. Since $n = 2k-z \geq 1$, $k$ satisfies $2k \geq z+1$. So $ \frac{(2k-z+5)^{2k-z}}{\ell^{6k}}$ is maximized when $2k = z + 1$. Therefore,
\begin{align*}
  (2k-z+5)^{2k-z}  \cdot \frac{((z+1)!)^6}{\ell^{6k}} 
  &\leq 6 \cdot \frac{((z+1)!)^6}{\ell^{6k}} = 6 \cdot \frac{((2k-n+1)!)^6}{\ell^{6k}} \tag{from definition of $z$}\\
  &\leq  6 \cdot \frac{\ell^{6(2k-n+1)}}{\ell^{6k}} \leq  \frac{6}{\ell^{6(n-k-1)}}\,. \tag{since $2k-n+1 < n \leq \ell$}
\end{align*}
Plugging this bound in~\eqref{eq:53} gives
\begin{align}
    \sum_{\sigma \in \mathcal{S}_\ell(n, k)} \prod_{i=1}^\ell b_{i, \sigma(i)} &\leq  \frac{1}{2\ell} \cdot \frac{6}{\ell^{6(n-k-1)}} \,. \label{eq:56}
\end{align}
Summing over all $n$ and $k$ gives
\begin{align}
    \sum_{n=2}^\ell \sum_{k=\lceil n/2 \rceil}^{n-2} \sum_{\sigma \in \mathcal{S}_\ell(n, k)} \prod_{i=1}^\ell b_{i, \sigma(i)} &\leq \sum_{n=2}^\ell \sum_{k=\lceil n/2 \rceil}^{n-2}  \frac{1}{2\ell} \cdot \frac{6}{\ell^{6(n-k-1)}}  \leq \sum_{n=2}^\ell \sum_{k=\lceil n/2 \rceil}^{n-2}  \frac{1}{\ell} \cdot \frac{3}{\ell^{6}} \tag{since $n-k-1 \geq 1$}\\
    \leq  \ell \cdot \ell \cdot  \frac{3}{\ell^{7}} \leq \frac{3}{\ell^5} \,.  \label{eq:57}
\end{align}

Plugging in~\eqref{eq:45},~\eqref{eq:49},~\eqref{eq:55}, and~\eqref{eq:57} into ~\eqref{eq:36}, 
\begingroup
\allowdisplaybreaks
\begin{align*}
  \perm(B_\ell) 
    {}={}& 1 + \sum_{n=2}^\ell \sum_{k=1}^{n-1} \sum_{\sigma \in \mathcal{S}_\ell(n, k)} \prod_{i=1}^\ell b_{i, \sigma(i)} \\
    {}={}& 1+ \sum_{n=2}^\ell \sum_{\sigma \in \mathcal{S}_\ell(n, 1)}\prod_{i=1}^\ell b_{i, \sigma(i)}  + \sum_{n=5}^\ell \sum_{k =2 }^{\lceil n/2\rceil-1} \sum_{\sigma \in \mathcal{S}_\ell(n, k)}\prod_{i=1}^\ell b_{i, \sigma(i)}\\
    &{}+{} \sum_{n=2}^\ell \sum_{\sigma \in \mathcal{S}_\ell(n, n-1)}\prod_{i=1}^\ell b_{i, \sigma(i)}+\sum_{n=2}^\ell \sum_{ k = \lceil n/2\rceil}^{n-2} \sum_{\sigma \in \mathcal{S}_\ell(n, k)}\prod_{i=1}^\ell b_{i, \sigma(i)} \\
    {}\leq{} & 1+\frac{1}{2\ell^5} + \frac{1}{240\ell^2}  + \frac{1}{2 \ell^6}+ \frac{3}{\ell^5} \leq 1 + \frac{0.05}{\ell}\,. \tag{since $\ell \geq 3$}
\end{align*}
Therefore $\perm(A_\ell) = \prod_{i=1}^\ell a_{i,i} \cdot \perm(B_\ell) \leq \prod_{i=1}^\ell a_{i,i} \cdot (1+\frac{0.05}{\ell}) $.
\end{proof}
\endgroup

\begin{definition}[Exceedances, Eulerian numbers, and Derangement numbers]
\label{lem:exc} \leavevmode
\begin{itemize}
    \item The exceedance of a permutation $\sigma \in \mathcal{S}_n$ is defined as $|\{i \in [n-1]: \sigma(i) > i\}|$.
    \item The Eulerian number $E(n, k)$ is defined to be the number of permutations in $\mathcal{S}_n$ with $k-1$ exceedances.
    \item The derangement number $T(n, k)$ is defined to be the number of derangements in $\mathcal{S}_n$ with $k$ exceedances. 
\end{itemize}
The explicit formula for $E(n, k)$ is  $ E(n,k)=\sum_{j=0}^{k+1}(-1)^{j}{\binom {n+1}{j}} \cdot (k+1-j)^{n}$ (page 273, \cite{comtet1974advanced}). The exponential generating function of $E(n,k)$ is given by (page 273, \cite{comtet1974advanced})
\begin{equation}
    \sum_{n=0}^{\infty} \sum_{k=0}^n E(n,k) \; t^{k} \; \frac{x^n}{n!} = \frac{t-1}{t-e^{(t-1)x}} \,. \label{eq:eulerian}
\end{equation}

The exponential generating function of $T(n,k)$ is given by (Proposition 5, \cite{brenti1990unimodal})
\begin{equation}
    \sum_{n=0}^{\infty} \sum_{k=1}^{n-1} T(n,k) \; t^{k} \; \frac{x^n}{n!} = e^{-t} \cdot\frac{t-1}{t-e^{(t-1)x}} \,. \label{eq:dr}
\end{equation}

Comparing~\eqref{eq:eulerian} and~\eqref{eq:dr}, we infer
\begin{equation}
    T(n,k) = \sum_{j=0}^{n} (-1)^{(n-j)} \cdot \binom{n}{j} \cdot E(j, k) \,. \label{eq:derangement}
\end{equation} 
\end{definition} 
\begin{lemma} \label{lem:derangements}
$\mathcal{S}_\ell(n,k)$ denote the set of permutations on $[\ell]$ with $\ell-n$ fixed points and $k$ exceedences. Then 
\begin{equation*}
    |\mathcal{S}_\ell(n,k) | \leq \binom{\ell}{n}\cdot \theta(n,k) \quad \text{ where } \quad \theta(n,k) := \begin{cases}1 & \text{ if } k = 1 \text{ or } k = n-1 \\
   (2k+3)^n & \text{ if } 2 \leq k < n/2 \\
    (2n-2k+5)^n & \text{ if } n/2 \leq k \leq n-2\,.
    \end{cases}
\end{equation*}
\end{lemma}
\begin{proof}
If we fix the $\ell-n$ fixed points for a permutation in $\mathcal{S}_\ell(n,k)$, then we get a derangement on the remaining $n$ points with $k$ exceedences. Therefore  $|\mathcal{S}_\ell(n,k)| = \binom{\ell}{n} \cdot T(n,k)$. 
So, it suffices to prove that \begin{itemize}
    \item $T(n,1) = T(n,n-1) = 1$, 
    \item $T(n, k) = T(n, n+1-k) < (2k+3)^{n}$ for any $k \in \{2, \ldots, n-2\}$.
\end{itemize}
From definition~\ref{lem:exc}, $ E(n,m)=\sum _{k=0}^{m+1}(-1)^{k}{\binom {n+1}{k}} \cdot (m+1-k)^{n}$. So removing terms with odd $k$ from $E(n,m)$ does not decrease it value. Therefore, 
\begin{align*}
    E(n, m) &< \binom{n+1}{0} \cdot (m+1)^{n} + \binom{n+1}{2} \cdot (m+1-2)^{n} + \ldots =  \sum_{i=0}^{m+1}{\binom {n+1}{2i}} \cdot (m+1-2i)^{n} \\
    &< (m+1)^n \cdot \left(\sum_{i = 0}^n \binom{n+1}{2i}  \right) \leq (m+1)^n \cdot 2^n.
\end{align*}
Using equation~\ref{eq:derangement} and taking absolute values, we get
\begin{align*}
    T(n,k) &< \sum_{j=0}^{n}  \binom{n}{j} \cdot E(j, k) < \sum_{j=0}^{n}  \binom{n}{j} \cdot  2^j \cdot (k+1)^j = (1 + 2(k+1))^{n}.
\end{align*}
\end{proof}
\end{document}